\let\ORIlabel\label
\let\ORIrefstepcounter\refstepcounter
\AddToHook{package/hyperref/before}{%
  \let\label\ORIlabel
  \let\refstepcounter\ORIrefstepcounter
}

\documentclass[hidelinks,onefignum,onetabnum]{siamart220329}

\usepackage{amsfonts,amssymb,bm,mathrsfs}
\usepackage{graphicx}
\usepackage{capt-of}
\DeclareGraphicsExtensions{.pdf,.png,.jpg}

\newcommand{\bbR}{\mathbb{R}}
\newcommand{\ket}[1]{| #1 \rangle}
\newcommand{\bra}[1]{\langle #1 |}

\newsiamremark{remark}{Remark}

\headers{Transmutation based Quantum Simulation}{S. Jin, C. Ma, and E. Zuazua}

\title{Transmutation based Quantum Simulation for Non-unitary Dynamics
}

\author{Shi Jin\thanks{
School of Mathematical Sciences, Shanghai Jiao Tong University,
Shanghai 200240, China;
Institute of Natural Sciences, Shanghai Jiao Tong University,
Shanghai 200240, China;
MOE-LSC, Shanghai Jiao Tong University,
Shanghai 200240, China. (\email{shijin-m@sjtu.edu.cn})}
	\and Chuwen Ma\thanks{
    School of Mathematical Sciences, East China Normal University,
Shanghai 200241, China;
Key Laboratory of MEA, Ministry of Education, East China Normal University, Shanghai 200241, China;
Shanghai Key Laboratory of PMMP, East China Normal University,
Shanghai 200241, China.
	(\email{cwma@math.ecnu.edu.cn})}
	\and Enrique Zuazua \thanks{Chair for Dynamics, Control, Machine Learning and Numerics, Alexander von Humboldt Professorship, Department of Mathematics, Friedrich-Alexander-Universit\"at
		Erlangen-N\"urnberg, 91058 Erlangen, Germany;
		Chair of Computational Mathematics, Fundaci\'on Deusto, Av. de las
		Universidades, 24, 48007 Bilbao, Basque Country, Spain;
		Departamento de Matem\'aticas, Universidad Aut\'onoma de Madrid,
		28049 Madrid, Spain (\email{enrique.zuazua@fau.de})
	}
	}

\ifpdf
\hypersetup{
    pdftitle={Transmutation based Quantum Simulation for Non-unitary Dynamics},
    pdfauthor={S. Jin, C. Ma, and E. Zuazua}
}
\fi

\begin{document}

	\maketitle

	\begin{abstract}
		We present a quantum algorithm for simulating dissipative diffusion dynamics generated by positive semidefinite operators of the form $A=L^\dagger L$, a structure that arises naturally in standard discretizations of elliptic operators. Our main tool is the Kannai transform, which represents the diffusion semigroup $e^{-AT}$, where $T$ is the final simulation time, as a Gaussian-weighted superposition of unitary wave propagators. For target accuracy $\varepsilon$, this representation leads to a linear-combination-of-unitaries implementation with a Gaussian tail and yields query complexity $\tilde{O}(\sqrt{\|A\|\,T\,\log(1/\varepsilon)})$, up to the standard dependence on state-preparation and output-norm factors, improving the scaling in $\|A\|,\,T\,$ and $\varepsilon$ compared with generic Hamiltonian-simulation-based methods. We instantiate the method for the heat equation and biharmonic diffusion under non-periodic physical boundary conditions, and further use it as a subroutine for constant-coefficient linear parabolic surrogates arising in entropy-penalization schemes for the viscous Hamilton--Jacobi equations.
        In the long-time regime, under a spectral-gap assumption, the same framework
         gives a structured quantum linear solver by exploiting convergence to the steady state. For normalized positive definite systems
        \(\sigma(A)\subset[1,\kappa]\), the solver outputs an \(\varepsilon\)-approximation to the state proportional to \(\bm x = A^{-1}\bm b\) with
        query complexity
        \(\widetilde O( \frac{\|\bm b\|}{\|\bm x\|}\sqrt{\kappa} \log^2\frac{\|\bm b\|}{\varepsilon\|\bm x\|} )\).
	\end{abstract}

	\begin{keywords}
		Kannai transform, Linear-combination-of-unitaries, Improved query complexity, Linear solver, Quantum algorithms for PDEs
	\end{keywords}

	\begin{MSCcodes}
	68Q12, 	65M12, 47D06
	\end{MSCcodes}

	\section{Introduction}
	Dissipative partial differential equations (PDEs), such as the heat equation and more general
	parabolic systems, are fundamental models for diffusion, relaxation, and transport across science
	and engineering. In high-dimensional regimes, classical discretization-based solvers may become
	computationally prohibitive, motivating alternative paradigms. Quantum computers can represent
	high-dimensional vectors as amplitude-encoded quantum states and, in principle, yield speedups for
	certain linear-algebraic primitives
	\cite{hidary2019quantum,NC2010Quantum,preskill2018quantum,rieffel2011quantum}.

	A core difficulty is that quantum hardware natively implements \emph{unitary} time evolution, whereas
	diffusion is intrinsically \emph{non-unitary}. This mismatch helps explain why Hamiltonian
	simulation---the simulation of Schr\"odinger dynamics---became an early and central target of quantum
	algorithms, leading to a mature toolkit with near-optimal complexity guarantees
	\cite{BCC15,Berry-Childs-Kothari-2015,BCC17,LC17,Low2019qubitization,CGJ19,Berry2019Dyson,JLY22SchrLong,An2022timedependent,fang2023time}.
	Designing quantum algorithms for dissipative dynamics therefore requires representing a non-unitary
	semigroup in terms of unitary primitives while retaining favorable dependence on the final time and
	the target accuracy.

	\medskip
	\noindent\textbf{Our approach.}
	We propose a quantum algorithm based on the classical \emph{Kannai transformation} (also known as
	\emph{transmutation} or \emph{subordination})
	\cite{Kannai77,SchillingSongVondracek12,ZuazuaSlides}.
	For generators of the form
	\[
	\mathcal{A}=\mathcal{L}^\dagger\mathcal{L}\succeq 0,
	\]
	the Kannai representation expresses the diffusion semigroup $e^{-t\mathcal{A}}$ as a Gaussian-weighted
	superposition of unitary wave propagators associated with $\mathcal{L}$. Algorithmically, we simulate
	the corresponding reversible dynamics and recover the diffusion solution by evaluating a Gaussian
	convolution in an auxiliary time variable.

	This viewpoint is related to dilation-based \emph{wrap-and-average} methods (e.g.\ Schr\"odingerization
	and its Fourier/contour variants), where the cost is governed by truncating an auxiliary kernel
	\cite{JLY24PRL,JLY22SchrLong,JLM24SchrInhom,ALL2023LCH,ACL2023LCH2,Li25,WLX25}.
	The key quantitative distinction is that exponentially decaying kernels typically require truncation
	radii $R=\Theta(T\log(1/\varepsilon))$, whereas the Gaussian Kannai kernel yields
	$R=\Theta(\sqrt{T \log(1/\varepsilon)})$, improving the $\varepsilon$-and-$T$-dependence of the resulting
	algorithms.

	\medskip
	\noindent\textbf{Related work.}
	Existing approaches to non-unitary simulation include:
	(i) discretize and solve via quantum linear system algorithms (QLSAs)
	\cite{Berry-2014,BerryChilds2017ODE,KroviODE,JLY2022multiscale,WL24,DLX25,Childs-Liu-2020,Berry2024Dyson};
	(ii) dilate to a unitary PDE  and apply Hamiltonian simulation
	\cite{JLY24PRL,JLY22SchrLong,JLM24SchrInhom}, with closely related variants such as LCHS
	\cite{ALL2023LCH,ACL2023LCH2} and contour-based matrix decomposition (CMD) \cite{WLX25};
	and (iii) treat non-unitary generators via block-encoding or Lindbladian techniques
	\cite{ShangGuoAnZhao25,fang2023time}.
	Our method belongs to the wrap-and-average family but exploits the Gaussian tail of the Kannai kernel
	to sharpen the dependence on $\varepsilon$ for the structured class
	$\mathcal{A}=\mathcal{L}^\dagger\mathcal{L}$.

	\medskip
	\noindent\textbf{Problem setting and contributions.}
	We consider inhomogeneous linear evolution problems
	\begin{equation}\label{eq:linear_eq}
		\frac{\mathrm{d}u}{\mathrm{d}t}=-\mathcal{A}u+f,
		\qquad u(0,\cdot)=u_0,
	\end{equation}
	where $\mathcal{A}=\mathcal{L}^\dagger\mathcal{L}$ is self-adjoint and non-negative and $f$ is a
	time-independent source. After spatial discretization, this yields a matrix $A=L^\dagger L$.
	We assume standard quantum access to the discretized operators via block-encodings and to the input
	data via state preparation (see Section~\ref{sec:quantum-implementation}).

	Our main contributions are as follows.
	\begin{itemize}
		\item \textbf{Kannai-based quantum simulation of diffusion.}
		We combine the Gaussian decay of the Kannai kernel with high-order quadrature in the auxiliary variable, the linear-combination-of-unitaries (LCU) framework, and near-optimal Hamiltonian simulation to obtain a quantum algorithm for \eqref{eq:linear_eq}.
		The query complexity to a block-encoding of the Hermitian dilation Hamiltonian scales as
		\[
		\widetilde{O}\!\bigl(u_r\,\sqrt{T\|A\|\log(1/\varepsilon)}\bigr),
		\]
		where $T$ is the final time, $\varepsilon$ is the target precision, and $u_r$ captures the relative
		size of the input and forcing. Up to polylogarithmic factors, this improves the dependence on
		$\varepsilon$, $T$, and $\|A\|$ compared with wrap-and-average methods based on exponentially
		decaying kernels.

		\item \textbf{PDE instantiations and boundary conditions.}
		We apply the framework to the heat equation and biharmonic diffusion, and show that Dirichlet and Neumann boundary conditions can be handled without changing the leading-order complexity under standard sparse discretizations on bounded domains.

     \item \textbf{A structured quantum linear solver from long-time dynamics.}
     In the strictly dissipative regime, we solve \(A\bm{x}=\bm{b}\) with
     \(A=L^\dagger L\) by identifying \(A^{-1}\bm b\) as the steady state of a forced
     dissipative evolution. After normalization
     \(\sigma(A)\subset[1,\kappa]\), the solver prepares the normalized state
     proportional to \(\bm x\) with query complexity
     \[
     \widetilde{O}\!\left(
     \frac{\|\bm b\|}{\|\bm x\|} \log\frac{1}{\varepsilon}
     \sqrt{\kappa \log\frac{\|\bm b\|}{\varepsilon \|\bm{x}\|}\log \frac{1}{\varepsilon} }\,
     \right).
     \]
     For smooth elliptic PDE data, the factor $\|\bm b\|/\|\bm x\|$ is often \(O(1)\); in the worst
     case it is bounded by \(\kappa\), yielding
     \(\widetilde{O}(\kappa^{3/2}\log^2(\kappa/\varepsilon))\).

		\item \textbf{Further transmutation constructions.}
		We collect additional transmutation-based reductions representing target dynamics as superpositions
		of reversible evolutions, including algorithms for second-order-in-time models (e.g.\ Euler--Poisson--Darboux via spherical means) and an alternative reduction from the heat equation to unitary transport via transport--heat averaging.
	\end{itemize}

	\medskip
	\noindent\textbf{Organization.}
	Section~\ref{subsec:transmutation} reviews transmutation and rewrites it as a first-order Hamiltonian
	system suitable for quantum simulation. Section~\ref{sec:numerical_dis} presents discretization,
	Gaussian quadrature, and error analysis. Section~\ref{sec:quantum-implementation} gives the quantum
	implementation and query bounds, and Section~\ref{sec:cost_of_exs} illustrates them on heat, viscous
	Hamilton--Jacobi, and biharmonic diffusion. Section~\ref{sec:longtime} addresses the long-time regime
	and the induced linear solver. Section~\ref{sec:other_transmutation} collects additional
	transmutation-based constructions.
	Section~\ref{sec:numerical-tests} presents numerical experiments that validate the correctness of the Kannai-transform--based LCU method.
	Conclusions are in Section~\ref{sec:conclusions}.

	\noindent\textbf{Notation.}
	For $j\in[N]$, $\ket{j}\in\mathbb{C}^N$ denotes the $j$th computational basis vector.
    For any nonzero vector \(\bm u\), \(\ket{\bm u}:=\bm u/\|\bm u\|_2\) denotes the associated normalized amplitude-encoded quantum state. Thus bold symbols denote unnormalized vectors, while kets denote normalized quantum states unless stated otherwise.
	We write $\widetilde{O}(\cdot)$ to suppress polylogarithmic factors, and use $O(\cdot)$
	with constants independent of the problem parameters of interest.

	\section{Transmutation and unitary superpositions}\label{subsec:transmutation}
    Transmutation methods provide a systematic way to represent one evolution problem
    as a superposition of another, often simpler, evolution whose analysis or numerical
    treatment is more accessible.
	We refer to
	\cite{CarrollShowalter1976,KatrakhovSitnik1984,Sitnik2012,Lions1956,Lions1959Weinstein,GriegoHersh1971}
	for representative developments.
	In this section we use transmutation techniques to express a target non-unitary
	dynamics as a superposition of unitary evolutions. This structure is
	particularly suitable for quantum algorithm design.

	\subsection{Kannai transformation for the linear system}

	We consider the setting in which the spatial operator admits the factorization
	\begin{equation}\label{eq:operator A}
		\mathcal{A} = \mathcal{L}^{\dagger}\mathcal{L},
	\end{equation}
	together with boundary conditions that make $\mathcal{A}$ nonnegative and
	self-adjoint.
	A prototypical example is $\mathcal{L}=-\nabla$ as a map from $H^1_0(\Omega)$ to
	$L^2(\Omega)$. In this case $\mathcal{A}=-\Delta_x$ with homogeneous Dirichlet
	boundary conditions.
	We study the linear system with a time-independent source term
	\begin{equation}\label{eq:diffu u}
		\frac{\mathrm{d} u (t,x)}{\mathrm{d} t} = -\mathcal{A}u(t,x) + f(x),
		\qquad u(0,x)=u_0(x).
	\end{equation}

	Kannai \cite{Kannai77} shows that one can represent $u$ via a transmutation
	operator acting only on the time variable:
	\begin{equation}\label{eq:Hw}
		u(t,x) = (\mathcal{H}w)(t,x)
		= \int_{\mathbb{R}} K(t,s)\,w(s,x)\,\mathrm{d}s.
	\end{equation}
	The kernel is Gaussian,
	\begin{equation}\label{eq:Gaussian kernel}
		K(t,s) = \frac{1}{\sqrt{4\pi t}}\,e^{-s^{2}/(4t)}.
	\end{equation}
	The auxiliary function $w$ solves the inhomogeneous wave-type equation
	\begin{equation}\label{eq:w-wave-A-LL}
		w_{ss}(s,x) = -\mathcal{A}w(s,x) + f(x),
		\qquad w(0,x)=u_0(x),\quad w_s(0,x)=0.
	\end{equation}

	\subsection{Unitary embedding via first-order factorization}
	To reveal the unitary dynamics underlying \eqref{eq:w-wave-A-LL}, we recast the
	second-order equation as a first-order system. Introduce an auxiliary variable
	$v$ by
	\[
	\frac{\mathrm{d}}{\mathrm{d}s}v(s,x) = -\mathcal{L}w(s,x),
	\qquad v(0,x)=0.
	\]
	Define the augmented state, the embedded source term, and the block operator by
	\begin{equation}\label{eq:def psi-wv}
		\bm{\psi}(s,x) = [\,w(s,x)\;\, v(s,x)\,]^{\top},
		\qquad
		\widetilde{\bm{f}}(x) = [\,f(x)\;\, 0\,]^{\top},
		\qquad
		\widetilde{\mathcal{L}}
		=
		\begin{bmatrix}
			0 & \mathcal{L}^{\dagger} \\
			-\mathcal{L} & 0
		\end{bmatrix}.
	\end{equation}
	A direct calculation shows that $\bm{\psi}$ satisfies the inhomogeneous
	first-order system
	\begin{equation}\label{eq:p}
		\frac{\mathrm{d}}{\mathrm{d}s}\bm{\psi}(s,x)
		=
		\widetilde{\mathcal{L}}\,\bm{\psi}(s,x)
		+ s\,\widetilde{\bm{f}}(x),
		\qquad
		\bm{\psi}(0,x)=\bm{\psi}_0(x),
	\end{equation}
	where $\bm{\psi}_0(x)=[\,u_0(x),\,0\,]^{\top}$.
	By Duhamel's principle, \eqref{eq:p} admits the variation-of-constants formula
	\begin{equation}\label{eq:unitary dynamics}
		\bm{\psi}(s,x)
		= \mathcal{U}(s)\,\bm{\psi}_0(x)
		+ \int_0^{s}\sigma\,\mathcal{U}(s-\sigma)\,\widetilde{\bm{f}}(x)\,
		\mathrm{d}\sigma,
		\qquad
		\mathcal{U}(s)=\exp\bigl(\widetilde{\mathcal{L}}s\bigr).
	\end{equation}
	The operator $\widetilde{\mathcal{L}}$ is skew-adjoint under the boundary
	conditions inherited from $\mathcal{L}$. Consequently, the homogeneous
	evolution $\mathcal{U}(s)$ is unitary. This unitary structure will be central
	in our quantum constructions.

	The next lemma makes explicit that the transmutation formula expresses $u$
	as an integral superposition of the unitary dynamics generated by
	$\widetilde{\mathcal{L}}$.

	\begin{lemma}\label{lem:transmutation-inhom-damped-wv}
		Let $\mathcal{A}=\mathcal{L}^{\dagger}\mathcal{L}$ be a nonnegative self-adjoint
		operator independent of time.
		Assume that $u$ and $w$ satisfy \eqref{eq:diffu u} and \eqref{eq:w-wave-A-LL},
		respectively.
		Let $\mathcal{H}$ be the transmutation operator in \eqref{eq:Hw} with kernel
		\eqref{eq:Gaussian kernel}.
		Then
		\begin{equation}\label{eq:u-from-w-short}
			u(t,x)
			=
			\Pi_1 \Bigg(
			\int_{\mathbb{R}} K(t,s)\,\mathcal{U}(s)\,\bm{\psi}_0(x)\,\mathrm{d}s
			+
			\int_{\mathbb{R}} K(t,s)\,
			\int_0^{s} \sigma\,\mathcal{U}(s-\sigma)\,\widetilde{\bm{f}}(x)\,
			\mathrm{d}\sigma\,\mathrm{d}s
			\Bigg),
		\end{equation}
		where $\Pi_1$ denotes the projection onto the first component, $\bm{\psi}_0$ and $\widetilde{\bm{f}}$ are defined in \eqref{eq:def psi-wv}, and $\mathcal{U}(s)$ is given by \eqref{eq:unitary dynamics}.
	\end{lemma}

	\section{Numerical discretization}\label{sec:numerical_dis}
	In this section we construct a fully discrete formulation of the proposed
	representation. We discretize the first-order system in space to obtain a semi-discrete unitary evolution, then approximate the resulting Gaussian integrals by quadrature and derive error bounds.

	\subsection{Spatial discretization of the first-order system}
	\label{subsec:Spatial discretization of the first-order system}
	To discretize \eqref{eq:p} in space, we introduce a matrix
	$L\in\mathbb{C}^{N\times N}$ arising from an appropriate spatial discretization
	of the differential operator $\mathcal{L}$.
	Let $\bm{w}_h(s),\bm{v}_h(s)\in\mathbb{C}^{N}$ denote the semi-discrete
	approximations of $w(s,\cdot)$ and $v(s,\cdot)$, and let
	$\bm{f}_h\in\mathbb{C}^{N}$ be the discrete source term.
	Define the stacked state vector and forcing vector by
	\begin{equation}\label{eq:psi_b}
		\bm{\psi}_h(s)=
		\begin{bmatrix}
			\bm{w}_h(s)\\[1mm]\bm{v}_h(s)
		\end{bmatrix}\in\mathbb{C}^{2N},
		\qquad
		\bm{\psi}_0=
		\begin{bmatrix}
			\bm{u}_0\\ \bm{0}
		\end{bmatrix}\in\mathbb{C}^{2N},
		\qquad
		\bm{b}=
		\begin{bmatrix}
			\bm{f}_h\\ \bm{0}
		\end{bmatrix}\in\mathbb{C}^{2N},
	\end{equation}
	where $\bm{u}_0$ is the discrete initial data. The semi-discrete system can
	be written as
	\begin{equation}\label{eq:discrete-system}
		\frac{\mathrm d}{\mathrm ds}\bm{\psi}_h(s)
		= \widetilde{L}\bm{\psi}_h(s) + s\bm{b}
		=: -iH\,\bm{\psi}_h(s)+s\bm{b},
	\end{equation}
	with the anti-Hermitian matrix $\widetilde{L}$ and its Hermitian counterpart $H$
	given by
	\begin{equation}\label{eq:H-Sch-short}
		\widetilde{L}= \ket{0}\bra{1}\otimes L^{\dagger} - \ket{1}\bra{0}\otimes L
		\in\mathbb{C}^{2N\times 2N},
		\qquad
		H=i\widetilde{L}.
	\end{equation}
	Hence the homogeneous propagator $U(s):=\exp(-iHs)$ is unitary. By Duhamel's
	principle, the solution of \eqref{eq:discrete-system} admits the representation
	\begin{equation}\label{eq:psi-duhamel}
		\bm{\psi}_h(s)
		=
		U(s)\bm{\psi}_0
		+\left(\int_0^s (s-\sigma)\,U(\sigma)\,\mathrm d\sigma\right)\bm{b}.
	\end{equation}

	Based on \eqref{eq:psi-duhamel}, we define the discrete target at time
	$T>0$ by the Kannai Gaussian integral
	\begin{align}
		\bm{u}_f^{\mathrm{disc}}(T)
		&:=\int_{\mathbb{R}} K(T,s)
		\,\bm{\psi}_h(s)\,\mathrm ds \notag  \\
		&=
		\big(\int_{\mathbb{R}}\kappa_T(s)\,U(s)\,\mathrm ds\big)\bm{\psi}_0
		+\big(\int_{\mathbb{R}}\kappa_T(s)
		\int_0^{s} (s-\sigma)\,U(\sigma)\,\mathrm d\sigma\,\mathrm ds\big)\bm{b}
		\label{eq:I1I2}
	\end{align}
	with
	\begin{equation}\label{eq:kappa}
		\kappa_T(s)=K(T,s) =(4\pi T)^{-1/2}e^{-s^2/(4T)}.
	\end{equation}
	We also introduce the Gaussian tail function
	\[
	\Phi_T(a)
	:=\frac{1}{\sqrt{4\pi T}}\int_a^\infty e^{-s^2/(4T)}\,\mathrm ds
	=\tfrac12\,\mathrm{erfc}\!\left(\frac{a}{2\sqrt{T}}\right).
	\]
	Since $U(s)$ is unitary and $\kappa_T$ is integrable on $\mathbb{R}$, all
	integrals above are absolutely convergent, and Fubini's theorem applies.
	Exchanging the order of integration gives the one-dimensional representation.
	More precisely, one obtains
	\begin{equation}\label{eq:vf_disc}
		\bm{u}_f^{\mathrm{disc}}(T)
		=
		\left(\int_{\mathbb{R}} \kappa_T(s)\,U(s)\,\mathrm ds\right)\bm{\psi}_0
		+\left(\int_{\mathbb{R}} \Lambda_T(\sigma)\,U(\sigma)\,\mathrm d\sigma\right)\bm{b},
	\end{equation}
	where the scalar kernel $\Lambda_T$ is given by

	\begin{equation}\label{eq:Lamda}
		\Lambda_T(\sigma)=
		\begin{cases}
			\displaystyle
			\frac{1}{\sqrt{4\pi T}}\int_{\sigma}^{\infty} e^{-s^2/(4T)}(s-\sigma)\,\mathrm ds
			= \sqrt{\frac{T}{\pi}}\,e^{-\sigma^2/(4T)}-\sigma\,\Phi_T(\sigma),
			& \sigma\ge 0,\\[10pt]
			\displaystyle
			\frac{1}{\sqrt{4\pi T}}\int_{-\infty}^{\sigma} e^{-s^2/(4T)}(\sigma-s)\,\mathrm ds
			= \sqrt{\frac{T}{\pi}}\,e^{-\sigma^2/(4T)}+\sigma\,\Phi_T(-\sigma),
			& \sigma<0.
		\end{cases}
	\end{equation}
    It follows directly from this formula that \(\Lambda_T\) is nonnegative and even.  Moreover,
	$ \int_{\bbR}\Lambda_T(\sigma)\,d\sigma =T$.

	\subsection{Numerical integration}
	For numerical implementation, we truncate the infinite integration domain to a
	finite interval $[-R,R]$ with $R>0$. We partition $[-R,R]$ into uniform panels
	of length $h_1>0$. Let $M_R:=R/h_1$ and assume that $M_R$ is an integer. The panels
	are indexed by $m=-M_R,\dots,M_R-1$ and given by
	$[mh_1,(m+1)h_1]$.
	On each panel we apply a $Q$-point Gauss--Legendre rule. Let
	$\{x_q,\omega_q\}_{q=1}^{Q}$ be the standard Gauss--Legendre nodes and weights
	on $[-1,1]$. The corresponding nodes and rescaled weights on
	$[mh_1,(m+1)h_1]$ are
	\begin{equation}\label{eq:weight_and_points}
		s_{q,m}=\frac{h_1}{2}x_q+\frac{2m+1}{2}h_1,
		\qquad
		w_{q,m}=\frac{h_1}{2}\omega_q,
		\qquad q=1,\dots,Q.
	\end{equation}

	We assume that the kernel values can be precomputed offline to high precision.
	More precisely, we use approximations $\kappa_T^a$ and $\Lambda_T^a$ satisfying
	\begin{equation}\label{eq:kappa_Lambda_approxi}
		|\kappa_T^a(s)-\kappa_T(s)|\le \delta_{\mathrm{off}}\,|\kappa_T(s)|,\quad
		|\Lambda_T^a(s)-\Lambda_T(s)|\le \delta_{\mathrm{off}}\,|\Lambda_T(s)|,
		\quad \forall s\in\mathbb{R}.
	\end{equation}
	With these ingredients, the quadrature approximations of the two integrals in
	\eqref{eq:vf_disc} lead to the discrete representation
	\begin{equation}\label{eq:uf_h}
		\bm{u}_f^{\mathrm{disc}}(T)\approx \bm{u}_f^{h}(T)
		=
		\sum_{m=-M_R}^{M_R-1}\sum_{q=1}^{Q}
		\Bigl(
		c_{q,m}\,U(s_{q,m})\,\bm{\psi}_0
		+d_{q,m}\,U(s_{q,m})\,\bm{b}
		\Bigr),
	\end{equation}
	where
	\begin{equation}\label{eq:coeffs}
		c_{q,m}=w_{q,m}\,\kappa_T^a(s_{q,m}),
		\qquad
		d_{q,m}=w_{q,m}\,\Lambda_T^a(s_{q,m}).
	\end{equation}

	\begin{remark}
		The offline precomputation of $\kappa_T$ and $\Lambda_T$ removes the need to
		evaluate nested integrals during the online stage. Both contributions in
		\eqref{eq:vf_disc} reduce to one-dimensional weighted sums over quadrature
		nodes. Alternative quadrature rules such as the trapezoidal rule or Monte Carlo
		sampling can be incorporated analogously. The choice of quadrature affects the
		required coefficient precision but does not change the final query complexity
		of the quantum procedure, since coefficient precision enters only
		polylogarithmically in the relevant qubit  complexity.
	\end{remark}

	\subsection{Error analysis for the fully discrete quadrature}

	We estimate the error of the fully discrete quadrature scheme approximating
	$\bm u_f^{\mathrm{disc}}(T)$ by $\bm u_f^{h}(T)$:
	\[
	\mathcal E
	= \|\bm u_f^{h}(T)-\bm u_f^{\mathrm{disc}}(T)\|.
	\]
	The error consists of three contributions:
	(i) truncation of the integral to $[-R,R]$,
	(ii) Gauss--Legendre quadrature discretization on each panel,
	and (iii) offline approximation of the kernel functions.

	The following theorem provides a complete error bound together with explicit
	parameter choices.

	\begin{theorem}\label{thm:error-fully-discrete-finish}
		Assume that the offline kernels $\kappa_T^a$ and $\Lambda_T^a$ satisfy the relative
		accuracy condition \eqref{eq:kappa_Lambda_approxi}.
		Let $\bm u_f^{h}(T)$ denote the truncated quadrature approximation
		\eqref{eq:uf_h} and $\bm u_f^{\mathrm{disc}}(T)$ the target representation
		\eqref{eq:vf_disc}.
		Then, for any $\epsilon\in(0,1)$, there exist parameters
		$R>0$, $Q\ge1$, $h_1>0$, and $\delta_{\mathrm{off}}>0$ such that
		\[
		\|\bm u_f^{h}(T)-\bm u_f^{\mathrm{disc}}(T)\|
		\le
		\epsilon\bigl(\|\bm\psi_0\|+T\|\bm b\|\bigr).
		\]
		Moreover, the total coefficient magnitudes satisfy
		\[
		\sum_{m,q}|c_{q,m}|\le 1+\epsilon,
		\qquad
		\sum_{m,q}|d_{q,m}|\le (1+\epsilon)T+\epsilon .
		\]

		In particular, the above bounds hold under the explicit choices
		\[
		R
		=2\sqrt{T\log\!\Big(\frac{8}{\epsilon}\Big)}
		=O\!\big(\sqrt{T\log(1/\epsilon)}\big),
		\]
		and
		\[
		Q=O\big(\log(1/\epsilon)\big),
		\qquad
		h_1\le
		\frac{1}{e\bigl(\|L\|+1/\sqrt{2T}\bigr)},
		\qquad
		\delta_{\mathrm{off}}\le\epsilon/4 .
		\]
	\end{theorem}

	\begin{proof}

	We decompose the error into truncation, quadrature, and offline approximation
	errors. Define the exact and truncated integral operators by
	\[
	\begin{aligned}
		\mathcal K_T
		&:=\int_{\mathbb R}\kappa_T(s)U(s)\,\mathrm ds,
		&
		\mathcal M_T
		&:=\int_{\mathbb R}\Lambda_T(s)U(s)\,\mathrm ds,\\
		\mathcal K_{T,R}
		&:=\int_{-R}^{R}\kappa_T(s)U(s)\,\mathrm ds,
		&
		\mathcal M_{T,R}
		&:=\int_{-R}^{R}\Lambda_T(s)U(s)\,\mathrm ds .
	\end{aligned}
	\]
	The corresponding composite Gauss--Legendre approximations with exact kernels are
	\[
		\mathcal K_{T,R}^{Q}
		:=\hspace{-4mm}\sum_{\hspace{-1mm}m=-M_R}^{M_R-1}\sum_{q=1}^{Q}
		w_{q,m}\kappa_T(s_{q,m})U(s_{q,m}),\quad
		\mathcal M_{T,R}^{Q}
		:=\hspace{-4mm}\sum_{\hspace{-1mm}m=-M_R}^{M_R-1}\sum_{q=1}^{Q}
		w_{q,m}\Lambda_T(s_{q,m})U(s_{q,m}),
	\]
	and the offline-kernel versions are
	\[
		\mathcal K_{T,R}^{Q,a}
		:=\hspace{-4mm}\sum_{\hspace{-1mm}m=-M_R}^{M_R-1}\sum_{q=1}^{Q}
		w_{q,m}\kappa_T^a(s_{q,m})U(s_{q,m}),\quad
		\mathcal M_{T,R}^{Q,a}
		:=\hspace{-4mm}\sum_{m=-M_R}^{M_R-1}\sum_{q=1}^{Q}
		w_{q,m}\Lambda_T^a(s_{q,m})U(s_{q,m}).
	\]
	Then
	\[
	\bm u_f^{h}(T)-\bm u_f^{\mathrm{disc}}(T)
	=
	(\mathcal K_{T,R}^{Q,a}-\mathcal K_T)\bm\psi_0
	+
	(\mathcal M_{T,R}^{Q,a}-\mathcal M_T)\bm b .
	\]
	It remains to bound the two operator differences
	\(\mathcal K_{T,R}^{Q,a}-\mathcal K_T\) and
	\(\mathcal M_{T,R}^{Q,a}-\mathcal M_T\). For instance,
	\[
	\|\mathcal K_{T,R}^{Q,a}-\mathcal K_T\|
	\le
	\|\mathcal K_T-\mathcal K_{T,R}\|
	+
	\|\mathcal K_{T,R}-\mathcal K_{T,R}^{Q}\|
	+
	\|\mathcal K_{T,R}^{Q}-\mathcal K_{T,R}^{Q,a}\|,
	\]
	and the same decomposition holds for the \(\mathcal M\)-term.

	\paragraph{Truncation}
	Since \(U(s)\) is unitary, the homogeneous truncation error is bounded by the
	Gaussian tail. With \(x=R/(2\sqrt T)\),
	\[
	\|\mathcal K_T-\mathcal K_{T,R}\|
	\le
	\int_{|s|>R}\kappa_T(s)\,\mathrm ds
	=
	\operatorname{erfc}(x)
	\le
	\frac{e^{-x^2}}{\sqrt{\pi}x}.
	\]
	For the inhomogeneous kernel, formula \eqref{eq:Lamda} gives
	\(\Lambda_T\ge0\), \(\Lambda_T\) even, and
	$
	\int_{\mathbb R}\Lambda_T(\sigma)\,\mathrm d\sigma=T.
	$
	Moreover, by Fubini's theorem,
	\[
	\int_{|\sigma|>R}\Lambda_T(\sigma)\,\mathrm d\sigma
	=
	\int_R^\infty \kappa_T(s)(s-R)^2\,\mathrm ds
	\le
	\frac{T}{\sqrt{\pi}x^3}e^{-x^2}.
	\]
	Thus, for
	\[
	R=2\sqrt{T\log(8/\epsilon)},
	\]
	up to changing the universal numerical constant in the logarithm, the truncation
	contribution is bounded by
	\[
	O(\epsilon)\bigl(\|\bm\psi_0\|+T\|\bm b\|\bigr).
	\]

	\paragraph{Quadrature}
	Let
	\[
	B:=\|H\|+\frac1{\sqrt{2T}}
	=
	\|L\|+\frac1{\sqrt{2T}}.
	\]
	On each panel, \(\kappa_T(s)U(s)\) is analytic. The kernel \(\Lambda_T\) is
	analytic on the two half-lines, and the point \(0\) is a panel endpoint, so the
	same composite Gauss--Legendre estimate applies branchwise. The standard
	Gauss--Legendre remainder estimate, together with
	\(\|\partial_s^rU(s)\|\le \|H\|^r\) and the Gaussian derivative bounds, yields
	\[
	\|\mathcal K_{T,R}-\mathcal K_{T,R}^{Q}\|
	\le
	Ce^{-cQ},
	\qquad
	\|\mathcal M_{T,R}-\mathcal M_{T,R}^{Q}\|
	\le
	CTe^{-cQ},
	\]
	provided
	\[
	h_1\le\frac{1}{eB}
	=
	\frac{1}{e(\|L\|+1/\sqrt{2T})}.
	\]
	Here \(C,c>0\) are universal constants independent of \(T\), \(\|L\|\), and
	\(\epsilon\). Choosing
	\[
	Q=C_0\log(1/\epsilon)
	\]
	with \(C_0\) sufficiently large makes the quadrature contribution
	\[
	O(\epsilon)\bigl(\|\bm\psi_0\|+T\|\bm b\|\bigr).
	\]

	\paragraph{Offline kernel approximation}
	By the relative accuracy assumption \eqref{eq:kappa_Lambda_approxi} and the
	unitarity of \(U(s)\),
	\[
	\|\mathcal K_{T,R}^{Q}-\mathcal K_{T,R}^{Q,a}\|
	\le
	\delta_{\mathrm{off}}
	\sum_{m,q}w_{q,m}\kappa_T(s_{q,m}),
	\quad
	\|\mathcal M_{T,R}^{Q}-\mathcal M_{T,R}^{Q,a}\|
	\le
	\delta_{\mathrm{off}}
	\sum_{m,q}w_{q,m}\Lambda_T(s_{q,m}).
	\]
	The scalar versions of the preceding truncation and quadrature estimates imply
	\[
	\sum_{m,q}w_{q,m}\kappa_T(s_{q,m})\le 1+C\epsilon,
	\qquad
	\sum_{m,q}w_{q,m}\Lambda_T(s_{q,m})\le T(1+C\epsilon)+C\epsilon .
	\]
	Hence, with \(\delta_{\mathrm{off}}\le\epsilon/4\), after absorbing universal
	constants into the parameter choices, the offline contribution is also bounded
	by
	\[
	O(\epsilon)\bigl(\|\bm\psi_0\|+T\|\bm b\|\bigr).
	\]
	Combining the three estimates gives
	\[
	\|\bm u_f^{h}(T)-\bm u_f^{\mathrm{disc}}(T)\|
	\le
	\epsilon\bigl(\|\bm\psi_0\|+T\|\bm b\|\bigr).
	\]

	Finally, we bound the coefficient magnitudes. Since the quadrature weights are
	positive and \(\kappa_T,\Lambda_T\ge0\),
	\[
	\sum_{m,q}|c_{q,m}|
	\le
	(1+\delta_{\mathrm{off}})
	\sum_{m,q}w_{q,m}\kappa_T(s_{q,m})
	\le
	1+\epsilon.
	\]
	Similarly,
	\[
	\sum_{m,q}|d_{q,m}|
	\le
	(1+\delta_{\mathrm{off}})
	\sum_{m,q}w_{q,m}\Lambda_T(s_{q,m})
	\le
	(1+\epsilon)T+\epsilon.
	\]
	This proves the stated error and coefficient bounds.
		\end{proof}

	\section{Quantum implementation and query complexity}\label{sec:quantum-implementation}
	In this section we present the quantum implementation of our method for the
	inhomogeneous system in \eqref{eq:diffu u} and derive the resulting query
	complexity bounds.

	For the quantum implementation we flatten the double index $(q,m)$ into a single
	index $j\in\{0,\dots,M-1\}$, where $M=2M_RQ$. We write
	\[
	c_j := c_{q,m},
	\qquad
	d_j := d_{q,m},
	\qquad
	s_j := s_{q,m}.
	\]
	Since the same quadrature nodes are used for both sums, the discrete
	representation \eqref{eq:uf_h} becomes
	\begin{equation}\label{eq:uf-disc-flat-short}
		\bm{u}_f^h
		=
		\sum_{j=0}^{M-1}
		\bigl(
		c_j\,U(s_j)\,\bm{\psi}_0
		+
		d_j\,U(s_j)\,\bm{b}
		\bigr),
	\end{equation}
	where $U(s)=\exp(-iHs)$ and $H$ is given in \eqref{eq:H-Sch-short}.

	\subsection{Block-encoding of the Hamiltonians}
	We work in the standard sparse-query/block-encoding model \cite{gilyen2019quantum,CGJ19,Lin2022Notes}.
	Assume access to an $(\alpha_L,a_L,0)$-block-encoding of $L$, namely a unitary $U_L$ such that
	\begin{equation}\label{eq:L-block-short}
		(\bra{0^{a_L}}\otimes I_N)\,U_L\,(\ket{0^{a_L}}\otimes I_N)=\frac{L}{\alpha_L},
		\qquad \alpha_L\ge \|L\|.
	\end{equation}
	Then $U_L^\dagger$ block-encodes $L^\dagger$ with the same parameters.
	We emphasize that access to a factor $L$ is natural for standard staggered or mixed discretizations, where $L$ corresponds to a discrete gradient/flux operator, and is not an additional assumption beyond assembling the PDE operator  $A = L^{\dagger} L$.

	Define the Hermitian dilation
	\[
	\widetilde{H}_0:=\ket{0}\!\bra{1}\otimes L^\dagger+\ket{1}\!\bra{0}\otimes L .
	\]
	Using one additional ``position'' qubit and $U_L,U_L^\dagger$, one can construct a unitary $W_{H_0}$
	acting on $a_H:=a_L+1$ ancillas and the $2N$-dimensional signal space such that
	\begin{equation}\label{eq:H0-block-final-short}
		(\bra{0^{a_H}}\otimes I_{2N})\,W_{H_0}\,(\ket{0^{a_H}}\otimes I_{2N})
		=\frac{\widetilde{H}_0}{\alpha_L}.
	\end{equation}
	With $P:=\ket{0}\!\bra{0}-i\ket{1}\!\bra{1}$, one has $H=(P\otimes I_N)\widetilde{H}_0(P^\dagger\otimes I_N)$.
	Define
	\begin{equation}\label{eq:UH-def-correct}
		\mathrm{HAM}_H:=(I_{a_H}\otimes P\otimes I_N)\,W_{H_0}\,(I_{a_H}\otimes P^\dagger\otimes I_N).
	\end{equation}
	Then $\mathrm{HAM}_H$ is an $(\alpha_H,a_H,0)$-block-encoding of $H$ with $\alpha_H=\Theta(\alpha_L)$,
	and each call to $\mathrm{HAM}_H$ uses $O(1)$ queries to $U_L$ and $U_L^\dagger$.

	Let $s_{\max}:=\max_j |s_j|\le R$ and set $\gamma_j:=s_j/s_{\max}\in[-1,1]$.
	Using reversible arithmetic, we implement scalar-preparation oracles $O_{\gamma,\mathrm{L}}$ and
	$O_{\gamma,\mathrm{R}}$ such that, conditioned on $\ket{j}$, they prepare amplitudes and phases for $\gamma_j$:
	\begin{equation}\label{eq:gamma-oracle}
		\begin{aligned}
			O_{\gamma,\mathrm{L}}:\ \ket{j}\ket{0}_\gamma &\mapsto
			\ket{j}\!\left(\sqrt{|\gamma_j|}\ket{0}_\gamma+\sqrt{1-|\gamma_j|}\ket{1}_\gamma\right),\\
			O_{\gamma,\mathrm{R}}:\ \ket{j}\ket{0}_\gamma &\mapsto
			\ket{j}\!\left(e^{i\theta_j^{(\gamma)}}\sqrt{|\gamma_j|}\ket{0}_\gamma+\sqrt{1-|\gamma_j|}\ket{1}_\gamma\right).
		\end{aligned}
	\end{equation}
	The resulting overhead is polylogarithmic in the target precision.
	With these, define the unitary $\mathrm{HAM}_{H_s}$ by
	\begin{equation}\label{eq:U-scalar-def}
		\begin{aligned}
			\mathrm{HAM}_{H_s}
			:={}&\, I_M \otimes \Bigl(
			\bigl(O_{\gamma,\mathrm{L}}^\dagger \otimes I_{a_H}\otimes I_{2N}\bigr) \\
			&\cdot
			\bigl(
			\ket{0}\!\bra{0}_\gamma\otimes \mathrm{HAM}_H
			+
			\ket{1}\!\bra{1}_\gamma\otimes I_{a_H}\otimes I_{2N}
			\bigr)\cdot
			\bigl(O_{\gamma,\mathrm{R}}\otimes I_{a_H}\otimes I_{2N}\bigr) \Bigr).
		\end{aligned}
	\end{equation}
	Then, on the joint index--signal space,
	\begin{equation}\label{eq:Hj-block-encoding-compact}
		(I_M\!\otimes\!\bra{0}_\gamma\bra{0^{a_H}}\!\otimes\! I_{2N})\,
		\mathrm{HAM}_{H_s}\,
		(I_M\!\otimes\!\ket{0}_\gamma\ket{0^{a_H}}\!\otimes\! I_{2N})
		=\sum_{j=0}^{M-1}\ket{j}\!\bra{j}\otimes \frac{s_j H}{\alpha_H s_{\max}} .
	\end{equation}
	Hence $\mathrm{HAM}_{H_s}$ is an $(\alpha_H s_{\max},a_H\!+\!1,0)$-block-encoding of the block-diagonal Hamiltonian
	\begin{equation}\label{eq:K}
		H_S:=\sum_{j=0}^{M-1} s_j\,\ket{j}\!\bra{j}\otimes H,
	\end{equation}
	where the index register is treated as part of the signal space and not counted as ancillas.

	\subsection{Implementation of  unitary time evolution}
	Using the above $(\alpha_H s_{\max},a_H+1,0)$-block-encoding of $H_S$, we aim to implement
	\[
	e^{-iH_S}=\sum_{j=0}^{M-1}\ket{j}\!\bra{j}\otimes e^{-iHs_j}=\sum_{j=0}^{M-1}\ket{j}\!\bra{j}\otimes U(s_j).
	\]
	Applying a QSVT-based Hamiltonian simulation routine to the normalized block-encoding
	$H_S/(\alpha_H s_{\max})$ \cite{Low2019Interaction,ACL2023LCH2}, we obtain a unitary $\mathrm{SEL}$
	acting on an additional ancilla register of size $a_{\mathrm{QSVT}}$ such that
	\begin{equation}\label{VaErr}
		(I_M\otimes \bra{0^{a_{\mathrm{tot}}}}\otimes I_{2N})\,\mathrm{SEL}\,
		(I_M\otimes \ket{0^{a_{\mathrm{tot}}}}\otimes I_{2N})
		=\sum_{j=0}^{M-1}\ket{j}\!\bra{j}\otimes U_j^a,
	\end{equation}
	where $a_{\mathrm{tot}}:=a_H+1+a_{\mathrm{QSVT}}$ and
	\[
	\|U_j^a-U(s_j)\|\le \delta_1 \qquad \text{for all } j\in\{0,\dots,M-1\}.
	\]
	By \cite[Corollary~16]{ACL2023LCH2}, constructing $\mathrm{SEL}$ uses the block-encoding oracle for $H_S$
	a total of
	\begin{equation}\label{eq:times-for-H}
		\widetilde{O}\!\Big(\alpha_H s_{\max}+\log(1/\delta_1)\Big)
	\end{equation}
	queries, up to polylogarithmic factors in the
	relevant precision parameters.

	\subsection{The LCU implementation of the discrete representation}
	We assume access to state-preparation oracles
	\[
	O_{\psi}:\;\ket{\bm 0}\mapsto\ket{\bm{\psi}_0},
	\qquad
	O_{b}:\;\ket{\bm 0}\mapsto\ket{\bm{b}}.
	\]
	We also assume access to standard coefficient oracles
	$O_{c,\mathrm{L}},O_{c,\mathrm{R}}$ and $O_{d,\mathrm{L}},O_{d,\mathrm{R}}$ that,
	conditioned on $j$, prepare amplitudes proportional to $|c_j|$ and $|d_j|$ and
	encode the phases $e^{i\theta_j^{(c)}}$ and $e^{i\theta_j^{(d)}}$ required by the
	linear combination of unitary (LCU) framework \cite{CKS17}.

	Write $c_j=|c_j|e^{i\theta_j^{(c)}}$ and $d_j=|d_j|e^{i\theta_j^{(d)}}$, and set
	\begin{equation}\label{eq:alpha_cd}
		\alpha_c := \sum_{j=0}^{M-1} |c_{j}|,
		\qquad
		\alpha_d := \sum_{j=0}^{M-1} |d_{j}|.
	\end{equation}
	Using \eqref{VaErr}, the homogeneous contribution in \eqref{eq:uf-disc-flat-short}
	admits the LCU block-encoding
	\begin{equation}\label{eq:block-hom-short}
		\mathcal{W}_{\mathrm{hom}}^a
		=
		(O_{c,\mathrm{L}}^\dagger\otimes I)\,
		\mathrm{SEL}\,
		(O_{c,\mathrm{R}}\otimes O_{\psi}),
		\qquad
		(\bra{0}\otimes I)\,
		\mathcal{W}_{\mathrm{hom}}^a\,
		(\ket{0}\otimes I)
		=
		\frac{1}{\alpha_c}
		\sum_{j=0}^{M-1}
		c_j\,U_j^a.
	\end{equation}
	The inhomogeneous contribution is realized analogously by
	\begin{equation}\label{eq:block-inh-short}
		\mathcal{W}_{\mathrm{inh}}^a
		=
		(O_{d,\mathrm{L}}^\dagger\otimes I)\,
		\mathrm{SEL}\,
		(O_{d,\mathrm{R}}\otimes O_{b}),
		\qquad
		(\bra{0}\otimes I)\,
		\mathcal{W}_{\mathrm{inh}}^a\,
		(\ket{0}\otimes I)
		=
		\frac{1}{\alpha_d}
		\sum_{j=0}^{M-1}
		d_j\,U_j^a.
	\end{equation}

	Summarizing, there exists an integer $n_a$ and a unitary $V_0$ acting on the index register,
	$n_a$ ancilla qubits, the position qubit $p$, and the $N$-dimensional data register such that
	\begin{equation}\label{eq:V0-summary}
		\ket{0^{\otimes r}}_{\mathrm{idx}}
		\ket{0^{\otimes n_a}}_{\mathrm{anc}}
		\ket{0}_{p}
		\ket{\bm 0}_{\mathrm{data}}
		\xrightarrow{\,V_0\,}
		\frac{1}{\eta_0}\,
		\ket{0^{\otimes r}}_{\mathrm{idx}}
		\ket{0^{\otimes n_a}}_{\mathrm{anc}}
		\otimes \bm{u}_f^{a}
		+\ket{\perp}.
	\end{equation}
	For brevity we omit additional workspace qubits used by reversible arithmetic, which are uncomputed at the end.
Here $\ket{\perp}$ denotes a state whose index--ancilla component is orthogonal to
$\ket{0^{\otimes r}}_{\mathrm{idx}}\ket{0^{\otimes n_a}}_{\mathrm{anc}}$.
The vector $\bm{u}_f^a$ lives on the joint $p$--data register. With the embedding
$\bm{\psi}_0=\ket{0}_p\otimes \bm{u}_0$, it is given by
\begin{equation}\label{eq:uf_a}
	\bm{u}_f^a
	=
	\sum_{j=0}^{M-1}
	\bigl(
	c_j U_j^a \bm{\psi}_0
	+
	d_j U_j^a \bm{b}
	\bigr),
	\qquad
	\eta_0
	=
	\Theta\bigl(\alpha_c\|\bm{u}_0\|+\alpha_d\|\bm{b}\|\bigr).
\end{equation}

	For clarity, we focus on the homogeneous case and present the overall quantum circuit that implements the corresponding LCU construction (see Fig.~\ref{fig:circuit}).
	The circuit starts from
	$\ket{0^{\otimes r}}_{\mathrm{idx}}
	\ket{0^{\otimes n_a}}_{\mathrm{anc}}
	\ket{0}_{p}
	\ket{0}_{\mathrm{data}},$
	where $r=\lceil \log_2 M\rceil$.
The data register is initialized by the state-preparation oracle $O_{\psi}$ so that
$ O_{\psi}\ket{\bm 0}_{\mathrm{data}}=\ket{\bm{u}_0}_{\mathrm{data}}, $
while the position qubit remains in the state $\ket{0}_{p}$. Equivalently, the
joint $p$--data register is initialized as
$\ket{\bm{\psi}_0} = \ket{0}_{p}\ket{\bm{u}_0}_{\mathrm{data}}. $
The ancilla register in Fig.~\ref{fig:circuit} represents the collection of auxiliary qubits used by the coefficient-preparation oracles and by \(\mathrm{SEL}\).

	\begin{figure}[t]
		\centering
		\includegraphics[scale=0.9]{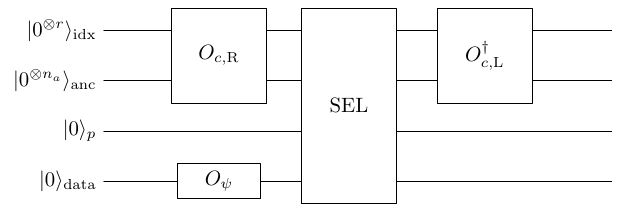}
		\caption{Quantum circuit for the homogeneous term.}
		\label{fig:circuit}
	\end{figure}

\subsection{Query complexity}\label{subsec:complexity}
Let \(\Pi_h:=\bra0_p\otimes I_N\). Since \(\bm u_f^a\) in
\eqref{eq:V0-summary}--\eqref{eq:uf_a} already denotes the all-zero
index--ancilla branch on the joint \(p\)-data register, the implemented
physical branch is
\(\bm u_h^a:=\Pi_h\bm u_f^a\). The target physical solution is
\(\bm u_h(T):=\Pi_h\bm u_f^{\mathrm{disc}}(T)\), with
\(\bm u_f^{\mathrm{disc}}(T)\) defined in \eqref{eq:vf_disc}.

The circuit \(V_0\) prepares the unnormalized augmented vector \(\bm u_f^a\)
with prefactor \(1/\eta_0\). Thus the joint success probability for obtaining
the physical branch is \(\|\bm u_h^a\|^2/\eta_0^2\). In the proof below, the
implementation errors are chosen so that
\(\|\bm u_h^a-\bm u_h(T)\|=O(\varepsilon\|\bm u_h(T)\|)\), and hence
\(\|\bm u_h^a\|=\Theta(\|\bm u_h(T)\|)\). Applying oblivious amplitude
amplification, the repetition factor satisfies
\begin{equation}\label{eq:repeat-times-new}
	g
	=
	O\!\left(\frac{\eta_0}{\|\bm u_h^a\|}\right)
	=
	O\!\left(\frac{\eta_0}{\|\bm u_h(T)\|}\right)
	=
	O\!\left(
	\frac{\alpha_c\|\bm u_0\|+\alpha_d\|\bm b\|}
	{\|\bm u_h(T)\|}
	\right).
\end{equation}
 	Define
	\begin{equation}\label{eq:vr-def-new}
		u_r
		:=
		\frac{\|\bm{u}_0\|+ T\|\bm{b}\|}{\|\bm{u}_h(T)\|}.
	\end{equation}
	The resulting query complexity bounds are stated in Theorem~\ref{thm:complexity}.

	\begin{theorem}\label{thm:complexity}
		Consider the linear equation in \eqref{eq:diffu u} with
		 $\mathcal{A}=\mathcal{L}^{\dagger}\mathcal{L}$ self-adjoint.
		Let $L$ be the matrix obtained from the spatial discretization of $\mathcal{L}$.
		Assume access to an $(\alpha_L,a_L,0)$-block-encoding of $L$ with $\alpha_L = \Theta(\|L\|)$.
		There exists a quantum algorithm that prepares an $\varepsilon$-approximation of
		the normalized solution state $\ket{\bm{u}_h(T)}$ with $\Omega(1)$ success
		probability and a success flag.
		Moreover, the total query complexity to the block-encoding oracle for $L$ is
		\[
		\widetilde{O}\!\left(
		u_r\Bigl(
		\|L\|\sqrt{T\,\log\!\frac{u_r}{\varepsilon}}
		+
		\log\!\frac{u_r}{\varepsilon}
		\Bigr)
		\right),
		\]
		and the number of calls to the state-preparation oracles $O_{\psi}$ and $O_b$ is
		$O(u_r)$.

		The above guarantees hold under the following default choices of auxiliary parameters.

		\begin{center}
            \begin{minipage}{\linewidth}
            \captionof{table}{Default discretization choices used in the proof.}
            \label{tab:disc-params}
            \centering
            \small
            \renewcommand{\arraystretch}{1.05}
            \begin{tabular}{lll}
				\hline
				Parameter & Meaning & Default choice  \\
				\hline
				$R$ & truncation radius in $s$ &
				$R = 2\sqrt{T\,\log\!\bigl(\dfrac{8u_r}{\varepsilon}\bigr)}$ \\
				$h_1$ & panel length in $s$ (step size) &
				$h_1 = \dfrac{1}{e\bigl(\|L\|+1/\sqrt{2T}\bigr)}$ \\
				$Q$ & reference Gauss points per panel &
				$Q = O\!\bigl(\log\!\bigl(\dfrac{8u_r}{\varepsilon}\bigr)\bigr)$ \\
				$\delta_{\mathrm{off}}$ & offline approximation error in \eqref{eq:kappa_Lambda_approxi} &
				$\delta_{\mathrm{off}} = \dfrac{\varepsilon}{32u_r}$ \\
				\hline
			\end{tabular}
            \end{minipage}
		\end{center}
	\end{theorem}

	\begin{proof}
	    Let \(\bm{u}_f^{\mathrm{disc}}\) be the exact semi-discrete Kannai target in
	    \eqref{eq:vf_disc}, let \(\bm{u}_f^{h}\) be the corresponding truncated
	    quadrature target in \eqref{eq:uf_h}, and let \(\bm{u}_f^{a}\) be the
	    unnormalized vector produced by the quantum circuit in \eqref{eq:uf_a}.
	    Since \(\Pi_h\) is a contraction, the augmented-vector errors imply
	    \[
	    \|\bm u_h^a-\bm u_h(T)\|
	    \le
	    \|\bm u_f^a-\bm u_f^h(T)\|
	    +
	    \|\bm u_f^h(T)-\bm u_f^{\mathrm{disc}}(T)\|.
	    \]
	    For any nonzero vectors \(\bm{x}\) and \(\bm{y}\),
	    \[
	    \left\|
	    \frac{\bm{x}}{\|\bm{x}\|}
	    -
	    \frac{\bm{y}}{\|\bm{y}\|}
	    \right\|
	    \le
	    \frac{2\|\bm{x}-\bm{y}\|}{\|\bm{x}\|}.
	    \]
	    Applying this with \(\bm{x}=\bm u_h(T)\) and
	    \(\bm{y}=\bm u_h^a\) yields
	    \begin{equation}\label{eq:total-error-split-new}
	    	\left\|
	    	\ket{\bm u_h^a}
	    	-
	    	\ket{\bm u_h(T)}
	    	\right\|
	    	\le
	    	\frac{2\|\bm u_f^{\mathrm{disc}}(T)-\bm u_f^h(T)\|}
	    	{\|\bm u_h(T)\|}
	    	+
	    	\frac{2\|\bm u_f^h(T)-\bm u_f^a\|}
	    	{\|\bm u_h(T)\|}.
	    \end{equation}

		\paragraph{Step 1: Discretization error}
		Choose the discretization parameters $(R,h_1,Q,\delta_{\mathrm{off}})$ as in
		Table~\ref{tab:disc-params}, obtained by applying
		Theorem~\ref{thm:error-fully-discrete-finish} with target accuracy $\varepsilon/(8u_r)$.
		Then
		\[
		\|\bm u_f^{\mathrm{disc}}(T)-\bm u_f^h(T)\|
		\le
		\frac{\varepsilon}{8u_r}
		\bigl(\|\bm u_0\|+T\|\bm b\|\bigr)
		=
		\frac{\varepsilon}{8}\|\bm u_h(T)\|.
		\]
		Thus the first term in \eqref{eq:total-error-split-new} is at most $\varepsilon/4$.

		\paragraph{Step 2: Error from imperfect implementation of unitaries}
		Let $U(s)=e^{-iHs}$ and let $U_j^a$ be the unitaries implemented by the selector
		$\mathrm{SEL}$, satisfying $\|U_j^a-U(s_j)\|\le \delta_1$ for all $j$ as in
		\eqref{VaErr}. Using the LCU representations of $\bm{u}_f^{h}$ and $\bm{u}_f^{a}$ yields
		\begin{align}\label{eq:HS-LCU-error-new}
			\|\bm{u}_f^{h}-\bm{u}_f^{a}\|
			=
			\big\|
			\sum_{j=0}^{M-1} c_j\bigl(U(s_j)-U_j^a\bigr)\bm{\psi}_0
			+
			\sum_{j=0}^{M-1} d_j\bigl(U(s_j)-U_j^a\bigr)\bm{b}\big \|
			=  O\!\left(\delta_1 u_r\|\bm u_h(T)\|\right).
		\end{align}
		Choosing \(\delta_1=\Theta(\varepsilon/u_r)\) makes the second term in
		\eqref{eq:total-error-split-new} at most \(\varepsilon/4\).

		Combining the two steps gives
		\[
		\left\|
		\ket{\bm u_h^a}
		-
		\ket{\bm u_h(T)}
		\right\|
		\le \frac{\varepsilon}{2},
		\]
		hence in particular an \(\varepsilon\)-approximation of the normalized physical
		solution state.

		\paragraph{Step 3: Query complexity}
		Each invocation of $\mathrm{SEL}$ uses
		\[
		\widetilde{O}\!\left(\alpha_H s_{\max}+\log\frac{1}{\delta_1}\right)
		=
		\widetilde{O}\!\left(\|L\|\,R+\log\frac{u_r}{\varepsilon}\right)
		\]
		queries to the block-encoding of $H$, by \eqref{eq:times-for-H} and the
		choice \(\delta_1=\Theta(\varepsilon/u_r)\).
        Since the block-encoding of \(H\) is constructed from the block-encoding of \(L\) using only \(O(1)\) calls to \(L\) and \(L^\dagger\), we count this as the same order of queries to the block-encoding oracle for \(L\).
        With the parameter choices in
		Table~\ref{tab:disc-params}, one has $R=\Theta\bigl(\sqrt{T\log(u_r/\varepsilon)}\bigr)$.
		Postselecting the position qubit yields a state proportional to
		\(\bm u_h^a\), which has been shown above to give an
		\(\varepsilon\)-approximation of \(\ket{\bm u_h(T)}\) after normalization.
		The repetition factor is \(g=O(u_r)\) by \eqref{eq:repeat-times-new}.
		The number of calls to $O_{\psi}$ and $O_b$ is also $O(u_r)$
		since each attempt invokes them $O(1)$ times.
	\end{proof}

	\begin{remark}\label{rem:kernel-truncation}
		Theorem~\ref{thm:complexity} yields only polylogarithmic dependence on the target
		accuracy $\varepsilon$. Many quantum constructions for non-unitary linear
		dynamics approximate the propagator denoted by $\mathcal{S}(T,0)$ via a truncated unitary superposition
		\[
		\mathcal{S}(T,0)u_0 \approx \int_{-R}^{R} K(s)\,e^{-iHs}u_0\,\mathrm{d}s,
		\qquad
		2\int_R^\infty |K(s)|\,\mathrm{d}s \le \varepsilon,
		\]
		for which the dominant matrix-query cost is typically proportional to $R$ up to
		polylogarithmic factors. In our case, the Kannai kernel is Gaussian, giving
		\[
		R=\Theta \bigl(\sqrt{T\log(1/\varepsilon)}\bigr),
		\]
		and hence the dominant term in the query bound is
		$\widetilde{O}\!\bigl(u_r\,\|L\|\sqrt{T\log(1/\varepsilon)}\bigr)$.
		Moreover, since $A=L^\dagger L$, we have $\|L\|=\|A\|^{1/2}$, so the bound exhibits a
		square-root dependence on $\|A\|$ rather than a linear one, improving upon
		\cite{Low25LCHS,optimalSchro,ACL2023LCH2}. Fig.~\ref{fig:kernel-comparison}
		illustrates the truncation behavior of the corresponding kernels. In particular,
		our improved dependence can be summarized as
		\[
		\widetilde{O}\bigl(\sqrt{\alpha_A T\log(1/\varepsilon)}\bigr),
		\]
		as reported in Table~\ref{tab:comparison} for the homogeneous setting.
	\end{remark}

	\begin{figure}[!htb]
		\centering
		\includegraphics[scale=0.43]{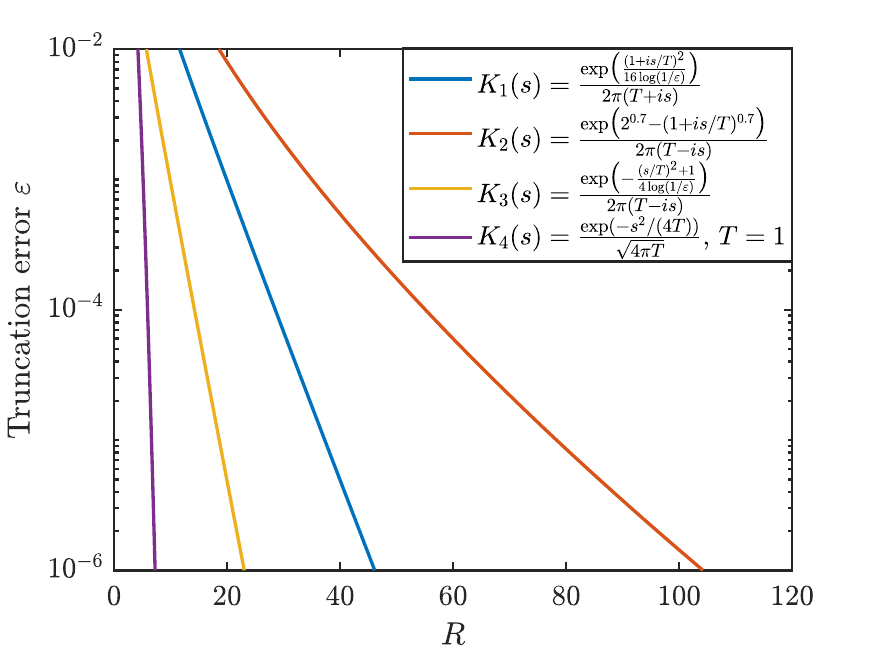}
		\includegraphics[scale=0.38]{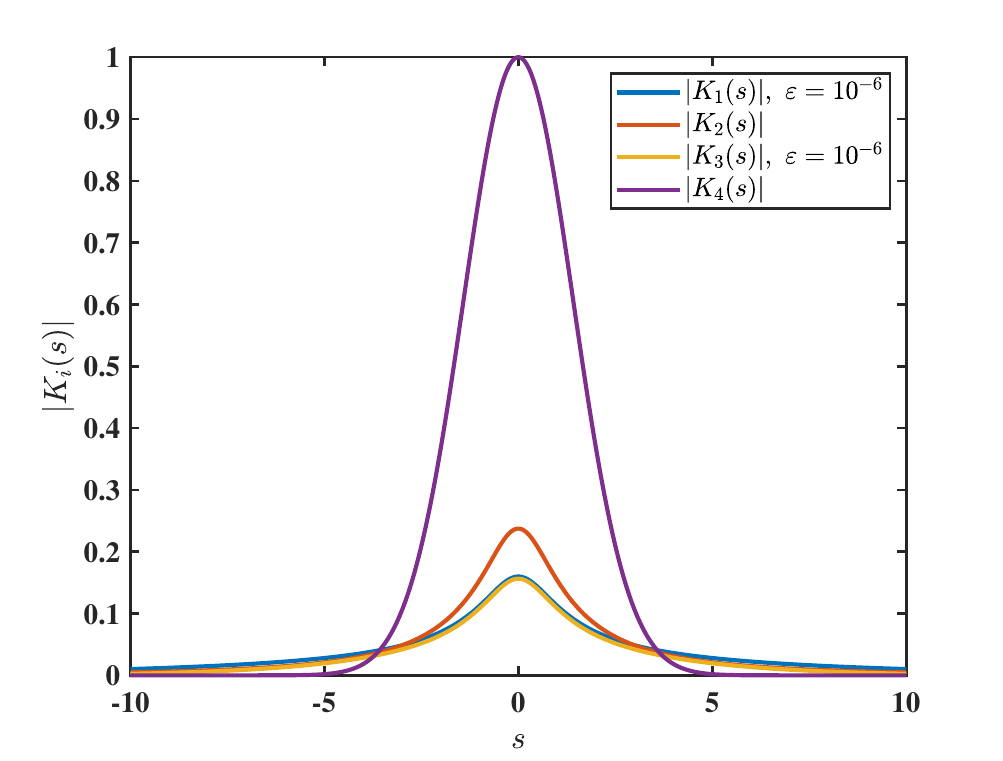}
		\caption{Kernel truncation behavior for several unitary-superposition representations.
		\textbf{Left:} truncation error $\varepsilon(R):=2\int_{R}^{\infty}\!\lvert K_i(s)\rvert\,\mathrm{d}s$
		  as a function of the truncation length $R$ (logarithmic scale in $\varepsilon$).
		\textbf{Right:} kernel magnitudes $\lvert K_i(s)\rvert$ on $s\in[-10,10]$ for a fixed target
		  precision $\varepsilon=10^{-6}$ and $T=1$.
		}
		\label{fig:kernel-comparison}
	\end{figure}

	\begin{table}[!htb]
		\centering
		\caption{Query complexities for homogeneous linear ODEs $\mathrm{d}\bm{u}/\mathrm{d}t=-A\bm{u}$
			with time-independent generator $A=L^{\dagger}L$. Here $T>0$ is the final time,
			$\alpha_A:=\|A\|$, $u_r:=\|\bm{u}_0\|/\|\bm{u}(T)\|$, and $\varepsilon\in(0,1)$ is the target precision.
			The parameter $\beta\in(0,1)$. }
		\label{tab:comparison}
		\renewcommand{\arraystretch}{1.25}
        {\small
		\begin{tabular}{c c c c}
			\hline\hline
			Method & Kernel $K(s)$
            & Truncation length
            & Matrix queries \\
			\hline
			\shortstack{Optimal\\Schr\"odingerization}
			\cite{optimalSchro}
			&
			$ \vspace{1mm} \displaystyle
			\frac{e^{\frac{(1+is^2/T^2)}{16\log(1/\varepsilon)}}}{2\pi (T +i s)}
            \vspace{1mm}
			$
			&
			$\displaystyle O\!\bigl(T\log\frac{1}{\varepsilon}\bigr)$
			&
			$\displaystyle
			\widetilde{O}\!\big(
			u_r\,
			\alpha_A\,T\,\log\frac{1}{\varepsilon}
			\big)$
			\\
			\hline
			\shortstack{Improved  LCHS
				\cite{ACL2023LCH2}}
			&\vspace{1mm}
			$\vspace{1mm} \displaystyle
			\frac{e^{2^{\beta}-(1+i \frac{s}{T})^{\beta}}}{2\pi(T-i s)}
			\vspace{1mm}$
			&
			$\displaystyle  O\!\bigl(T\log^{\frac{1}{\beta}}\frac{1}{\varepsilon}\bigr)$
			&
			$\displaystyle
			\widetilde{O}\!\big(
			u_r\,
			\alpha_A\,T\,\log^{\frac{1}{\beta}}\frac{1}{\varepsilon}
			\big)$
			\\
			\hline
			\shortstack{Optimal LCHS
				\cite{Low25LCHS}}
			&
			$\vspace{1mm} \displaystyle
			\frac{e^{-\frac{1+ s^2/T^2}{4\log(1/\varepsilon)}}}{2\pi(T-i s)}
			\vspace{1mm} $
			&
			$\displaystyle  O\!\bigl(T\log\frac{1}{\varepsilon}\bigr)$
			&
			$\displaystyle
			\widetilde{O}\!\big(
			u_r\,
			\alpha_A\,T\,\log\frac{1}{\varepsilon}
			\big)$
			\\
			\hline
			This work
			&
			$\vspace{1mm} \displaystyle
			\frac{e^{-\frac{s^2}{4T}}}{\sqrt{4\pi T}}
			\vspace{1mm} $
			&
			$\displaystyle  O\!\big(\sqrt{T\log\frac{1}{\varepsilon}}\big)$
			&
			$\displaystyle
			\widetilde{O}\!\big(
			u_r\,
			\sqrt{\alpha_A\,T\,\log\frac{1}{\varepsilon}}
			\big)$
			\\
			\hline\hline
		\end{tabular}
        }
	\end{table}

	\subsection{Extension beyond the factorized generator \texorpdfstring{$A=L^\dagger L$}{A = L-dagger L}}\label{subsec:generalA_commuting}
	In this subsection we extend the framework to time-independent generators $A\in\mathbb{C}^{N\times N}$
	that are not necessarily of the form $L^\dagger L$. For clarity, we focus on the homogeneous system
	\begin{equation}\label{eq:hom_generalA}
		\frac{\mathrm{d}}{\mathrm{d} t}\bm{u}(t)=-A\,\bm{u}(t),\qquad \bm{u}(0)=\bm{u}_0,
	\end{equation}
	since the inhomogeneous case can be handled analogously via Duhamel's principle.
	Throughout this subsection we omit auxiliary workspace qubits used by reversible arithmetic and Hamiltonian simulation,
	as they are uncomputed at the end.
	We use the Cartesian decomposition
	\begin{equation}\label{eq:cartesian_decomp}
		A = H_1 + iH_2,
		\qquad
		H_1:=\frac{A+A^\dagger}{2},
		\qquad
		H_2:=\frac{A-A^\dagger}{2i},
	\end{equation}
	where $H_1$ and $H_2$ are Hermitian. Assuming $H_1\succeq 0$, we may factor
	\begin{equation}\label{eq:H1_factorization}
		H_1=L^\dagger L
	\end{equation}
	for some matrix $L$.

\subsubsection{ 	 Normal case \texorpdfstring{$AA^{\dagger} = A^{\dagger} A$}{A A-dagger = A-dagger A}}
	If $A$ is normal, then $[H_1,H_2]=0$ and the propagator factors as
	\begin{equation}\label{eq:factorization_commuting}
		\bm u(T)=e^{-AT}\bm u_0
		=
		e^{-H_1T}\,e^{-iH_2T}\,\bm u_0
		=
		e^{-L^\dagger L\,T}\,e^{-iH_2T}\,\bm u_0.
	\end{equation}
	We therefore first implement the unitary $e^{-iH_2T}$ by Hamiltonian simulation and then apply the
	dissipative factor $e^{-L^\dagger L\,T}$ using the Kannai-transform LCU routine.

	Let $H$ be the Hermitian dilation associated with $L$ as in \eqref{eq:H-Sch-short} and define $U(s)=e^{-iHs}$.
	Set
	\begin{equation}\label{eq:kannai_commuting 0}
	U_{\kappa_T}:=\int_{\mathbb{R}}\kappa_T(s)\,U(s)\,\mathrm{d} s,
	\qquad
	\kappa_T(s)=(4\pi T)^{-1/2}e^{-s^2/(4T)}.
	\end{equation}
	After truncation and quadrature, $U_{\kappa_T}$ is realized by an LCU unitary $U_{\mathrm{LCU},T}$.
	Let $n_a$ denote the number of ancilla qubits used by this LCU block.

	Let the position qubit $p$ encode the $2\times 2$ block structure and initialize
	$\ket{\bm\psi_0}:=\ket{0}_p\ket{\bm u_0}_{\mathrm{data}}$.
	Then the Kannai representation reads
	\begin{equation}\label{eq:kannai_commuting}
		e^{-L^\dagger L\,T}\ket{\bm u_0}
		=
		(\bra{0}_p\otimes I)\,U_{\kappa_T}\,\ket{\bm\psi_0}.
	\end{equation}
	Combining \eqref{eq:factorization_commuting} and \eqref{eq:kannai_commuting}, we obtain
	\begin{equation}\label{eq:final_commuting_formula}
		\ket{\bm{u}(T)}
		\;\approx\;
		(\bra{0^{\otimes r}}_{\mathrm{idx}}\bra{0^{\otimes n_a}}_{\mathrm{anc}}\bra{0}_{p}\otimes I)\,
		U_{\mathrm{LCU},T}\,
		\ket{0^{\otimes r}}_{\mathrm{idx}}\ket{0^{\otimes n_a}}_{\mathrm{anc}}\ket{0}_{p}\,
		\bigl(e^{-iH_2T}\ket{\bm{u}_0}_{\mathrm{data}}\bigr).
	\end{equation}
	The query complexity is summarized
	in Corollary~\ref{lem:commuting_complexity}.
	\begin{corollary}\label{lem:commuting_complexity}
		Let $A\in\mathbb C^{N\times N}$ satisfy the decomposition
		\eqref{eq:cartesian_decomp} with $H_1\succeq 0$ and factorization $H_1=L^\dagger L$
		as in \eqref{eq:H1_factorization}, and assume $[H_1,H_2] =  0$.
		Assume (i) block-encoding access to $L$ as in \eqref{eq:L-block-short}, and
		(ii) Hamiltonian-simulation access to $H_2$, e.g., via an $(\alpha_{H_2},a,0)$ block-encoding oracle.
		Then the state $\ket{\bm u(T)}\propto e^{-AT}\ket{\bm u_0}$ can be prepared using
		\eqref{eq:final_commuting_formula} up to accuracy $\varepsilon$, with total query complexity
		\[
		\widetilde{O}\!\left(
		u_r\Bigl(
		\|L\|\sqrt{T\,\log\!\frac{u_r}{\varepsilon}}
		+\log\!\frac{u_r}{\varepsilon}
		+\alpha_{H_2}T
		\Bigr)
		\right),
		\qquad
		u_r := \frac{\|\bm{u}_0\|}{\|\bm{u}(T)\|},
		\]
		up to polylogarithmic factors. Moreover, the number of calls to the state-preparation
		oracles is $O(u_r)$.
	\end{corollary}

\subsubsection{Non-normal case \texorpdfstring{$AA^\dagger\neq A^\dagger A$}{A A-dagger unequal to A-dagger A}}
\label{subsubsec:noncommute_splitting}
In this case, equivalently \([H_1,H_2]\neq0\), the factorization
$e^{-(H_1+iH_2)T}=e^{-H_1T}e^{-iH_2T}$
is no longer available. Let \(\tau=T/N_t\). We use the Strang splitting--which has second-order accuracy:
\begin{equation}\label{eq:strang_noncommute}
	e^{-(H_1+iH_2)T}
	\approx
	\left(e^{-iH_2\tau/2}e^{-H_1\tau}e^{-iH_2\tau/2}\right)^{N_t}.
\end{equation}

Let $H$ be the Hermitian dilation associated with $L$ as in
\eqref{eq:H-Sch-short}, and set \(U(s)=e^{-iHs}\).  For one step of length
\(\tau\), define
\[
U_{\kappa_\tau}
:=
\int_{\mathbb R}\kappa_\tau(s)U(s)\,\mathrm ds,
\qquad
\kappa_\tau(s)=(4\pi\tau)^{-1/2}e^{-s^2/(4\tau)} .
\]
By the Kannai block identity \eqref{eq:kannai_commuting}, with \(T\)
replaced by \(\tau\) and \(L^\dagger L=H_1\), it holds
\begin{equation}\label{eq:kannai_block_tau}
	E_\tau
	:=
	(\bra0_p\otimes I)U_{\kappa_\tau}(\ket0_p\otimes I)
	=
	e^{-H_1\tau}.
\end{equation}

Let \(\{s_\ell(\tau),c_\ell(\tau)\}_{\ell=0}^{M-1}\) be the quadrature
nodes and LCU coefficients for the one-step kernel \(\kappa_\tau\), and
let \(U_\ell^a\) be the Hamiltonian-simulation approximation of
\(U(s_\ell(\tau))=e^{-iH s_\ell(\tau)}\).
Set
\begin{equation}\label{eq:U_kappa_tau_lcu0}
	\widetilde U_{\kappa_\tau}
	:=
	\sum_{\ell=0}^{M-1}c_\ell(\tau)U_\ell^a,
	\qquad
	\alpha_c(\tau):=\sum_{\ell=0}^{M-1}|c_\ell(\tau)|.
\end{equation}
Its physical dissipative block is
\[
\widetilde E_\tau
:=
(\bra0_p\otimes I)\widetilde U_{\kappa_\tau}(\ket0_p\otimes I)
\approx e^{-H_1\tau}.
\]
Set
\[
V_\tau:=I_p\otimes e^{-iH_2\tau/2},
\qquad
\widetilde{\mathcal T}_\tau
:=
V_\tau\widetilde U_{\kappa_\tau}V_\tau .
\]
Here $\widetilde{\mathcal T}_\tau$ acts on the enlarged space
$\mathcal H_p\otimes\mathcal H_{\rm data}$.  Its one-step physical block on the
data space is
\begin{equation}\label{eq:S_tau_physical}
	\widetilde{\mathscr{T}}_\tau
	:=
	(\bra0_p\otimes I)\widetilde{\mathcal T}_\tau(\ket0_p\otimes I)
	=
	e^{-iH_2\tau/2}\widetilde E_\tau e^{-iH_2\tau/2}.
\end{equation}

We do not postselect after each step.  Instead, we use a single global LCU for the
entire product.  Since
\[
\widetilde{\mathcal T}_\tau
=
\sum_{\ell=0}^{M-1}c_\ell(\tau)V_\tau U_\ell^aV_\tau,
\]
expanding the $N_t$-fold product gives
\begin{equation}\label{eq:global_path_expansion}
	(\widetilde{\mathcal T}_\tau)^{N_t}
	=
	\sum_{\boldsymbol\ell\in\{0,\ldots,M-1\}^{N_t}}
	c_{\boldsymbol\ell}(\tau)\,\mathcal U_{\boldsymbol\ell},
\end{equation}
where $\boldsymbol\ell=(\ell_1,\ldots,\ell_{N_t})$,
$
c_{\boldsymbol\ell}(\tau):=\prod_{k=1}^{N_t}c_{\ell_k}(\tau)$ ,
and
$\mathcal U_{\boldsymbol\ell}
:=
(V_\tau U_{\ell_{N_t}}^aV_\tau)\cdots(V_\tau U_{\ell_1}^aV_\tau)$.
The order of these factors is kept fixed, since the operators do not need to commute.

Let $U_{\rm path}$ be the corresponding path-LCU unitary.  Then
\begin{equation}\label{eq:path_block_full}
	(\bra{\mathbf0}_{\rm path}\otimes I_{p,{\rm data}})
	U_{\rm path}
	(\ket{\mathbf0}_{\rm path}\otimes I_{p,{\rm data}})
	=
	\frac{(\widetilde{\mathcal T}_\tau)^{N_t}}{\alpha_{\rm path}},
\end{equation}
with
\[
\alpha_{\rm path}
:=
\sum_{\boldsymbol\ell}|c_{\boldsymbol\ell}(\tau)|
=
\alpha_c(\tau)^{N_t}.
\]
Applying the coefficient bound in
Theorem~\ref{thm:error-fully-discrete-finish} to the one-step kernel
\(\kappa_\tau\), with target accuracy \(O(\varepsilon/(u_rN_t))\), gives
\[
\alpha_c(\tau)
\le
1+O\!\left(\frac{\varepsilon}{u_rN_t}\right).
\]
Consequently,
\[
\alpha_{\rm path}
=
\alpha_c(\tau)^{N_t}
=
O(1).
\]
Hence the LCU normalization does not introduce an exponential loss in \(N_t\).

The operator induced on the physical data space by the full path construction is
\begin{equation}\label{eq:S_path_def}
	\widetilde{\mathcal{T}}_{\rm path}(T)
	:=
	(\bra0_p\otimes I)(\widetilde{\mathcal T}_\tau)^{N_t}(\ket0_p\otimes I).
\end{equation}
Consequently,
\begin{equation}\label{eq:path_block_noncommute}
	(\bra{\mathbf0}_{\rm path}\bra0_p\otimes I)
	U_{\rm path}
	(\ket{\mathbf0}_{\rm path}\ket0_p\otimes I)
	=
	\frac{\widetilde{\mathcal{T}}_{\rm path}(T)}{\alpha_{\rm path}}.
\end{equation}
If the symmetric quadrature and Hamiltonian-simulation subroutines preserve the
$p=0$ subspace exactly, then
$\widetilde{\mathcal{T}}_{\rm path}(T)=\widetilde{\mathscr T}_{\tau}^{N_t}$;
otherwise the possible leakage is included in the step implementation error.  For a
normalized input, the success probability is
\[
p_{\rm succ}
=
\frac{\|\widetilde{\mathcal{T}}_{\rm path}(T)\bm u_0\|^2}
{\alpha_{\rm path}^2\|\bm u_0\|^2},
\]
so the final amplitude-amplification overhead is governed by the output-norm factor
$u_r$, not by an exponential factor in $N_t$.  The normalized output is
\begin{equation}\label{eq:final_noncommuting_formula}
	\ket{\bm u(T)}
	\approx
	\frac{
		(\bra{\mathbf0}_{\rm path}\bra0_p\otimes I)
		U_{\rm path}\ket{\mathbf0}_{\rm path}\ket0_p\ket{\bm u_0}
	}{
		\left\|
		(\bra{\mathbf0}_{\rm path}\bra0_p\otimes I)
		U_{\rm path}\ket{\mathbf0}_{\rm path}\ket0_p\ket{\bm u_0}
		\right\|
	}.
\end{equation}

\begin{corollary}\label{lem:noncommuting_complexity}
	Let $A\in\mathbb C^{N\times N}$ admit the decomposition
	\eqref{eq:cartesian_decomp} with $H_1\succeq0$ and factorization
	$H_1=L^\dagger L$ as in \eqref{eq:H1_factorization}, and assume
	$[H_1,H_2]\neq0$.  Define
	\[
	\Lambda
	:=
	\|[H_2,[H_2,H_1]]\|
	+
	\|[H_1,[H_1,H_2]]\| .
	\]
	For any $T>0$ and normalized-state accuracy $\varepsilon\in(0,1)$, choose
	\[
	N_t
	=
	\Theta\!\left(\sqrt{\frac{\Lambda T^3u_r}{\varepsilon}}\right),
	\qquad
	u_r:=\frac{\|\bm u_0\|}{\|\bm u(T)\|}.
	\]
	Assume block-encoding access to $L$ as in \eqref{eq:L-block-short} and
	Hamiltonian-simulation access to $H_2$, for instance via an
	$(\alpha_{H_2},a,0)$ block-encoding oracle.  Then the state
	$\ket{\bm u(T)}\propto e^{-AT}\ket{\bm u_0}$ can be prepared using
	\eqref{eq:final_noncommuting_formula} up to accuracy $\varepsilon$ with query
	complexity
	\[
	\widetilde{O}\!\left(
	u_r\left(
	\|L\|\left(\frac{\Lambda T^5u_r}{\varepsilon}\right)^{1/4}
	+
	\sqrt{\frac{\Lambda T^3u_r}{\varepsilon}}
	\log\frac{u_rN_t}{\varepsilon}
	+
	\alpha_{H_2}T
	\right)
	\right).
	\]
	  Moreover, the number of calls to the
	state-preparation oracle is $O(u_r)$.
\end{corollary}

\begin{proof}
	The Strang splitting error satisfies
	\[
	\left\|
	e^{-(H_1+iH_2)T}
	-
	\left(e^{-iH_2\tau/2}e^{-H_1\tau}e^{-iH_2\tau/2}\right)^{N_t}
	\right\|
	=
	O(\Lambda T\tau^2).
	\]
	For normalized states, it is enough to make the operator error
	$O(\varepsilon/u_r)$, using the standard inequality
	\[
	\left\|\frac{\bm x}{\|\bm x\|}-\frac{\bm y}{\|\bm y\|}\right\|
	\le
	\frac{2\|\bm x-\bm y\|}{\|\bm x\|}.
	\]
	This gives the stated choice of $N_t$.

	For each step, the Kannai LCU approximates $e^{-H_1\tau}$ in the physical block.
	The quadrature, Hamiltonian-simulation, and block-leakage errors are chosen at level
	$O(\varepsilon/(u_rN_t))$, so that their accumulated contribution is
	$O(\varepsilon/u_r)$.  The same choice gives
	$\alpha_{\rm path}=\alpha_c(\tau)^{N_t}=O(1)$, hence the final success probability
	is $\Omega(1/u_r^2)$ and amplitude amplification contributes a factor $O(u_r)$.

	One path-SELECT applies $N_t$ controlled one-step Kannai simulations.  The total
	Kannai cost per attempt is therefore
	\[
	\widetilde{O}\!\left(
	N_t\left(\|L\|\sqrt\tau+
	\log\frac{u_rN_t}{\varepsilon}\right)
	\right)
	=
	\widetilde{O}\!\left(
	\|L\|\left(\frac{\Lambda T^5u_r}{\varepsilon}\right)^{1/4}
	+
	\sqrt{\frac{\Lambda T^3u_r}{\varepsilon}}
	\log\frac{u_rN_t}{\varepsilon}
	\right).
	\]
	The two $H_2$ half steps over all Strang steps have total simulated time $T$, giving
	an additional $\widetilde{O}(\alpha_{H_2}T)$ cost.  Multiplication by the
	final amplification factor $O(u_r)$ yields the claimed complexity.
\end{proof}

\begin{remark}\label{rem:commute-vs-noncommute}
	If $[H_1,H_2]=0$ and $H_1=L^\dagger L$, the commuting construction of
	Corollary~\ref{lem:commuting_complexity} applies.  If $[H_1,H_2]\neq0$, the above
	global path LCU combines the same Kannai block for $e^{-H_1\tau}$ with Strang's
	splitting, while avoiding intermediate measurements.
\end{remark}

	\section{Complexity for representative PDEs}\label{sec:cost_of_exs}
	In this section, we instantiate the general complexity bound in
	Theorem~\ref{thm:complexity} for representative partial differential equations.
	We focus on how boundary conditions and spatial discretization determine the
	operator norm $\|L\|$, and hence the overall query complexity of the proposed
	quantum algorithm.

	\subsection{The heat equation}\label{subsubsec:heat}
	We begin with the heat equation on a bounded domain $\Omega\subset\mathbb{R}^d$,
	\[
	\partial_t u = -\mathcal{A}u + f,
	\qquad
	\mathcal{L}=-\nabla,
	\qquad
	\mathcal{L}^{\dagger}=\mathrm{div},
	\qquad
	u(0,x)=u_0(x),
	\]
	so that $\mathcal{A}=\mathcal{L}^\dagger\mathcal{L}=-\mathrm{div}\,\nabla$.
	Consider Dirichlet boundary data
	\[
	u(t,x)\big|_{x\in\partial\Omega} = u_d(x),
	\qquad t\ge 0.
	\]
	In the first-order embedding $\bm{\psi}=[\,w\;\;v\,]^\top$, the boundary
	condition is imposed on the auxiliary variable $w$ by requiring
	\[
	w(s,x)\big|_{x\in\partial\Omega}=u_d(x),
	\qquad \text{for all } s.
	\]
	Since the reconstruction uses the projection $\Pi_1\bm{\psi}=w$, the recovered
	solution $u(t,\cdot)$ inherits the prescribed Dirichlet values.

	We discretize the first-order system for $\bm{\psi}=[\,w\;\;v\,]^\top$ on a
	staggered grid such that the discrete gradient and divergence are represented by
	$-L$ and $L^\dagger$, respectively.

	In one dimension, let $\Omega=(0,1)$ and $h=1/N_x$. Place $w$ on primary nodes
	$x_i=ih$ and $v$ on midpoints $x_{i+\frac12}=(i+\tfrac12)h$. Collect
	\begin{align*}
		\bm{w}_h(s)&=(w_1(s),\ldots,w_{N_x-1}(s))^\top\in\mathbb{R}^{N_x-1},\\
		\bm{v}_h(s)&=(v_{1/2}(s),\ldots,v_{N_x-1/2}(s))^\top\in\mathbb{R}^{N_x},
	\end{align*}
	and enforce Dirichlet data by fixing $w_0(s)=u_d(x_0)$ and
	$w_{N_x}(s)=u_d(x_{N_x})$ for all $s$.
	Define the interior difference operator
	\(D_h\in\mathbb{R}^{N_x\times(N_x-1)}\) by
	\begin{equation}\label{eq:L-1D}
		(D_h\bm w_h)_{i+\frac12}
		:=
		\begin{cases}
			-\dfrac{w_1}{h}, & i=0,\\[0.7em]
			-\dfrac{w_{i+1}-w_i}{h}, & i=1,\ldots,N_x-2,\\[0.7em]
			\dfrac{w_{N_x-1}}{h}, & i=N_x-1 .
		\end{cases}
	\end{equation}
	For nonhomogeneous Dirichlet data, the boundary values
	\(w_0(s)=u_d(x_0)\) and \(w_{N_x}(s)=u_d(x_{N_x})\) are fixed for all \(s\).
	Equivalently,
	\[
	-\frac{w_{i+1}-w_i}{h}
	=
	(D_h\bm w_h)_{i+\frac12}
	+
	(\bm g_{\rm bd}^{(1)})_{i+\frac12},
	\]
	where
	\[
	\bm g_{\rm bd}^{(1)}
	=
	\frac1h
	\bigl(
	u_d(x_0),0,\ldots,0,-u_d(x_{N_x})
	\bigr)^\top
	\in\mathbb R^{N_x}.
	\]
	Thus \(D_h\bm w_h+\bm g_{\rm bd}^{(1)}\) is the boundary-corrected
	discrete approximation of \(-\nabla w\). In particular, for homogeneous
	Dirichlet data, \(\bm g_{\rm bd}^{(1)}=0\). Moreover,
	\[
	(\nabla w)_{i+\frac12}
	\approx
	-\bigl(D_h\bm w_h+\bm g_{\rm bd}^{(1)}\bigr)_{i+\frac12},
	\qquad
	(\operatorname{div}v)_i
	\approx
	(D_h^\top\bm v_h)_i .
	\]

	Let \(\bm f_h^{\rm phys}\in\mathbb R^{N_x-1}\) denote the interior
	discretization of the physical source term \(f\). Direct discretization of
	the wave equation with the prescribed boundary values gives
	\[
	\bm w_{h,ss}
	=
	-D_h^\top D_h\bm w_h+\bm f_h^{\rm eff},
	\qquad
	\bm f_h^{\rm eff}
	:=
	\bm f_h^{\rm phys}-D_h^\top\bm g_{\rm bd}^{(1)}
	\in\mathbb R^{N_x-1}.
	\]
	Hence the nonhomogeneous boundary data only modify the source term; the
	linear operator remains \(D_h^\top D_h\) with homogeneous interior unknowns.

	To match the abstract first-order formulation in Section~\ref{sec:numerical_dis}, we write
	\[
	\frac{d}{ds}
	\begin{bmatrix}
		\bm w_h\\
		\bm v_h
	\end{bmatrix}
	=
	\begin{bmatrix}
		0&D_h^\top\\
		-D_h&0
	\end{bmatrix}
	\begin{bmatrix}
		\bm w_h\\
		\bm v_h
	\end{bmatrix}
	+
	s
	\begin{bmatrix}
		\bm f_h^{\rm eff}\\
		0
	\end{bmatrix}.
	\]
	Indeed, differentiating the first component gives the second-order equation
	above. Thus, in the notation of the general formulation, the source vector is
	$
	\bm b=
	\begin{bmatrix}
		\bm f_h^{\rm eff}\\
		0
	\end{bmatrix},
	$
	and we relabel \(\bm f_h^{\rm eff}\) as \(\bm f_h\) when applying
	Theorem~\ref{thm:complexity}. We write \(L:=D_h\) for the interior difference operator in one
	dimension.

	In $d$ dimensions on $\Omega=(0,1)^d$, order interior nodal values of $w$ into
	$\bm{w}_h(s)\in\mathbb{R}^{M_w}$ with $M_w=(N_x-1)^d$, and collect staggered face
	values of $v$ into $\bm{v}_h(s)\in\mathbb{R}^{M_v}$. Lift $D_h$ to each
	coordinate direction by
	\begin{equation}\label{eq:Lk-kron}
		L^{(k)}=
		I_{N_x-1}^{\otimes (k-1)}\otimes D_h\otimes I_{N_x-1}^{\otimes (d-k)},
		\qquad k=1,\ldots,d,
	\end{equation}
	and stack the directional operators as
	\[
	L=
	\begin{bmatrix}
		(L^{(1)})^{\dagger} & (L^{(2)})^{\dagger} & \cdots & (L^{(d)})^{\dagger}
	\end{bmatrix}^{\dagger}.
	\]
	Then $L:\mathbb{R}^{M_w}\to\mathbb{R}^{M_v}$ discretizes $-\nabla$ on the
	staggered grid and $L^\dagger$ discretizes $\mathrm{div}$.

\begin{remark}[Rectangular discrete gradients]
	In the staggered discretization, the discrete gradient
	\(L:\mathbb C^{M_w}\to\mathbb C^{M_v}\) is generally rectangular. This causes
	no difficulty for the quantum implementation. Indeed, let
	\(N=\max\{M_w,M_v\}\), and let
	\[
	E_w:\mathbb C^{M_w}\to\mathbb C^N,
	\qquad
	E_v:\mathbb C^{M_v}\to\mathbb C^N
	\]
	be the canonical embeddings. The rectangular matrix \(L\) can be embedded into
	the zero-padded square matrix
	\[
	\widetilde L:=E_v L E_w^\dagger\in\mathbb C^{N\times N}.
	\]
	Then
	\[
	\|\widetilde L\|=\|L\|,
	\qquad
	\widetilde L^\dagger\widetilde L
	=
	E_w(L^\dagger L)E_w^\dagger
	\]
	on the embedded physical \(w\)-space. Moreover, the first-order block operator
	built from \(\widetilde L\) leaves the embedded physical subspace
	\(E_w\mathbb C^{M_w}\oplus E_v\mathbb C^{M_v}\) invariant and agrees there
	with the original rectangular formulation. Hence zero padding does not change
	the simulated solution, the relevant operator norm, or the query complexity.
	We therefore keep the rectangular notation \(L:\mathbb C^{M_w}\to\mathbb C^{M_v}\)
	in the PDE discretizations, with the understanding that this harmless padding
	can be used when a square block-encoding is required.
\end{remark}
		For nonhomogeneous Dirichlet data, the one-dimensional boundary correction
		described above is applied in each coordinate direction. This gives a stacked
		face-space vector \(\bm g_{\rm bd}\in\mathbb R^{M_v}\), so that
		\(L\bm w_h+\bm g_{\rm bd}\) is the boundary-corrected discretization of
		\(-\nabla w\). If \(\bm f_h^{\rm phys}\in\mathbb R^{M_w}\) denotes the
		interior discretization of the physical source, then the effective source is
		\[
		\bm f_h^{\rm eff}
		=
		\bm f_h^{\rm phys}-L^\dagger\bm g_{\rm bd}
		\in\mathbb R^{M_w}.
		\]
		Thus the abstract source vector is
		$
		\bm b=
		\begin{bmatrix}
			\bm f_h^{\rm eff}\\
			\bm 0_{M_v}
		\end{bmatrix},
		$
		and \(\bm f_h^{\rm eff}\) is relabeled as \(\bm f_h\) when applying
		Theorem~\ref{thm:complexity}.
	The semi-discrete first-order system is obtained from the continuum formulation by
	replacing $(\mathcal{L},\mathcal{L}^\dagger)$ with $(L,L^\dagger)$ and using the
	effective source above.

	\begin{remark}\label{remark:Neumann bd}
		For homogeneous Neumann data $\partial_n u=0$ on $\partial\Omega$, it is natural
		to impose the boundary condition on the flux component. One convenient staggered
		placement samples $v$ on primary nodes and $w$ on staggered locations, so that
		the normal component of $v$ can be fixed to zero on $\partial\Omega$. The
		resulting discrete gradient--divergence pair has the same structure as above,
		with the roles of $w$ and $v$ interchanged.
	\end{remark}

	The following corollary specializes Theorem~\ref{thm:complexity} to the heat
	equation under the staggered finite-difference discretization described above.

	\begin{corollary}\label{cor:heat-complexity}
		Consider $\Omega=(0,1)^d$ and the staggered finite-difference discretization
		above with mesh size $h$ in each coordinate direction. Let $L$ denote the
		resulting discrete gradient operator.
		Assume the hypotheses of Theorem~\ref{thm:complexity} hold.
		Then there exists a  quantum algorithm that prepares an
		$\varepsilon$-approximation of the normalized solution state $\ket{\bm{u}_h(T)}$
		approximating $u$
		with $\Omega(1)$ success probability and a success flag, using
		\[
		\widetilde{O}\!\left(
		u_r\left(
		\frac{\sqrt{d}}{h}\,\sqrt{T\log\frac{u_r}{\varepsilon}}
		+
		\log\frac{u_r}{\varepsilon}
		\right)
		\right)
		\]
		queries to the block-encoding oracle for $L$,
         where, in discrete norms,
		\[
		u_r=
		O\!\left(
		\frac{\|\bm u_{0,h}\|+T\|\bm f_h^{\rm eff}\|}
		{\|\bm u_h(T)\|}
		\right).
		\]
		Here \(\bm u_{0,h}\) is the interior discretization of the initial datum
		\(u_0\).
		Moreover, the number of calls to the state-preparation oracles $O_\psi$ and
		$O_b$ is $O(u_r)$.
	\end{corollary}

	\begin{proof}
		For the staggered finite-difference discretization on $(0,1)^d$, the discrete
		gradient operator satisfies $\|L\|=\Theta(\sqrt{d}/h)$. Substituting this bound
		into Theorem~\ref{thm:complexity} yields the stated query complexity. The bound
		on the number of calls to $O_\psi$ and $O_b$ follows from the same theorem.
	\end{proof}

	\subsection{Viscous Hamilton--Jacobi equations via entropy penalization}
	\label{subsec:HJ-entropy}

	Hamilton--Jacobi equations are fundamental in mechanics and optimal control.
	Since classical solutions may develop finite-time singularities, one often
	passes to vanishing-viscosity regularizations and viscosity solutions.  We use
	the entropy-penalization framework of~\cite{JinLiu25HJ,Gome07} only as a
	mechanism for producing a linear parabolic surrogate suitable for the
	Kannai-transform algorithm.  We do not reproduce the full entropy-penalization
	analysis here.

	For simplicity, we absorb the calibration constant in the entropy-penalization
	kernel into the viscosity parameter and write the effective viscosity again as
	\(\nu\).  Thus the viscous Hamilton--Jacobi equation is written as
	\begin{equation}\label{eq:HJ-visc}
		\partial_t S^\nu(t,\bm x)
		+
		H\!\bigl(\bm x,\nabla S^\nu(t,\bm x)\bigr)
		=
		\nu\,\Delta S^\nu(t,\bm x),
		\qquad
		S^\nu(0,\bm x)=S_0(\bm x),
		\qquad \bm x\in\mathbb T^d .
	\end{equation}

	The general entropy-penalization setting starts from a strictly convex kinetic
	energy \(K:\mathbb R^d\to\mathbb R\) and, possibly, a potential \(V(\bm x)\).
	The corresponding Hamiltonian may be written as
	\begin{equation}\label{eq:H-legendre}
		H(\bm x,\bm p)
		=
		\sup_{\bm v\in\mathbb R^d}
		\bigl\{-\bm p\cdot\bm v-K(\bm v)+V(\bm x)\bigr\}.
	\end{equation}

	The continuous-time limit of the entropy-penalized scheme gives a linear
	parabolic surrogate for an auxiliary variable \(u\) of the form
	\cite[Lemma~5]{JinLiu25HJ}
	\begin{equation}\label{eq:HJ-linear-drift-diff}
		\partial_t u(t,\bm x)
		=
		\bm\mu\cdot\nabla u(t,\bm x)
		+
		\nabla\!\cdot\!\bigl(\Sigma\nabla u(t,\bm x)\bigr),
		\qquad
		u(0,\bm x)=u_0(\bm x),
	\end{equation}
	where the drift vector \(\bm\mu\in\mathbb R^d\) and the diffusion matrix
	\(\Sigma\in\mathbb R^{d\times d}\) are constant in \(\bm x\) and \(t\).  More
	precisely, with
	\[
	\bm v_\star:=\arg\min_{\bm v\in\mathbb R^d}K(\bm v),
	\]
	one has
	\begin{equation}\label{eq:mu-nu-from-K}
		\bm\mu=\bm v_\star,
		\qquad
		\Sigma
		=
		\bm v_\star\bm v_\star^{\top}
		+
		2\nu\bigl(\nabla^2K(\bm v_\star)\bigr)^{-1}.
	\end{equation}

	The Cole-Hopf relation is
	\begin{equation}\label{eq:HJ-recover}
		S^\nu(t,\bm x)
		=
		-2\nu\log u(t,\bm x),
	\end{equation}
	up to an additive normalization constant.  We assume that this initial
	normalization is known classically; otherwise the procedure recovers
	\(S^\nu\) only up to an additive constant.  We also assume positivity of the
	surrogate solution, as holds for positive Cole-Hopf initial data under the
	heat semigroup.

	Since \(\bm\mu\) is constant, the drift in \eqref{eq:HJ-linear-drift-diff} can
	be removed by the translation
	\begin{equation}\label{eq:translation}
		\widetilde u(t,\bm x):=u(t,\bm x-\bm\mu t).
	\end{equation}
	Then \(\widetilde u\) solves
	\begin{equation}\label{eq:pure-diff}
		\partial_t\widetilde u(t,\bm x)
		=
		\nabla\!\cdot\!\bigl(\Sigma\nabla\widetilde u(t,\bm x)\bigr)
		=:-\mathcal A\,\widetilde u(t,\bm x).
	\end{equation}
	For a nonzero drift, the value \(u(T,\bm x_a)\) is recovered from
	\(\widetilde u(T,\bm x_a+\bm\mu T)\) on the torus.

	In the complexity estimate below, we further specialize to the isotropic
	constant case
	\[
	\bm\mu=0,
	\qquad
	\Sigma=\nu I_d,
	\]
	so that
	\[
	\mathcal A=-\nu\Delta .
	\]
	This is the heat-type surrogate to which we apply the general
	Kannai-transform state-preparation theorem.

	Let \(N_x\) be the number of grid points per direction and set \(h:=1/N_x\).
	On the periodic grid
	\[
	\mathcal G_h
	:=
	\{\,\bm x_{\bm n}=h\bm n:
	\bm n\in\{0,1,\dots,N_x-1\}^d\,\},
	\]
	we represent \(\widetilde u(t,\cdot)\) by the value vector
	\(\widetilde{\bm u}_h(t)\in\mathbb C^{N_x^d}\), defined by
	\[
	\bigl(\widetilde{\bm u}_h(t)\bigr)_{\bm n}
	=
	\widetilde u(t,\bm x_{\bm n}).
	\]
	Let \(U_F\) be the unitary discrete Fourier transform and define
	\[
	\widehat{\bm u}(t):=U_F\widetilde{\bm u}_h(t),
	\qquad
	\widetilde{\bm u}_h(t)=U_F^\dagger\widehat{\bm u}(t).
	\]
	We use the periodic second-order finite-difference discretization of
	\(-\nu\Delta\).  The Fourier modes are indexed by
	\[
	\bm k=(k_1,\dots,k_d)\in\mathcal K,
	\qquad
	\mathcal K
	:=
	\left\{
	-\frac{N_x}{2},-\frac{N_x}{2}+1,\dots,\frac{N_x}{2}-1
	\right\}^d ,
	\qquad N_x\ \text{even}.
	\]
	The semi-discrete generator is diagonalised as
	\begin{equation}\label{eq:A-diag}
		A=U_F^\dagger D U_F,
		\qquad
		L_{\mathrm{HJ}}:=D^{1/2}U_F,
		\qquad
		A=L_{\mathrm{HJ}}^\dagger L_{\mathrm{HJ}},
	\end{equation}
	where
	\begin{equation}\label{eq:FD-symbol}
		D(\bm k)
		=
		\nu\sum_{r=1}^d
		\frac{4}{h^2}
		\sin^2\!\left(\frac{\pi k_r}{N_x}\right).
	\end{equation}
	Consequently,
	\begin{equation}\label{eq:LHJ-norm}
		\|L_{\mathrm{HJ}}\|
		=
		\|D^{1/2}\|
		\le
		\frac{2\sqrt{d\nu}}{h}.
	\end{equation}

	Applying Theorem~\ref{thm:complexity} to the homogeneous semi-discrete system
	\[
	\frac{\mathrm d}{\mathrm dt}\widetilde{\bm u}_h(t)
	=
	-A\widetilde{\bm u}_h(t)
	\]
	gives the following state-preparation bound.

	\begin{corollary}
		\label{cor:HJ-complexity}
		Let \(A\) and \(L_{\mathrm{HJ}}\) be defined in \eqref{eq:A-diag}.  There exists
		a quantum algorithm that prepares an \(\varepsilon\)-approximation of the
		normalized amplitude state $\ket{\widetilde{\bm u}_h(T)}$
		with \(\Omega(1)\) success probability and query complexity
		\[
		\widetilde{O}\!\left(
		u_r
		\left(
		\frac{\sqrt{d\nu}}{h}
		\sqrt{T\log\!\left(\frac{u_r}{\varepsilon}\right)}
		+
		\log\!\left(\frac{u_r}{\varepsilon}\right)
		\right)
		\right),
		\qquad
		u_r:=
		\frac{\|\widetilde{\bm u}_h(0)\|}
		{\|\widetilde{\bm u}_h(T)\|}.
		\]
	\end{corollary}

	\subsubsection{Output model and accuracy for the Hamilton--Jacobi observable}

	Corollary~\ref{cor:HJ-complexity} prepares only the normalized amplitude state
	of the linear parabolic surrogate.  The Hamilton--Jacobi observable is
	obtained through the Cole-Hopf readout, not by full tomography.

	Let \(\bm u_h(T)\) denote the discrete vector after undoing any
	drift translation.  In the isotropic case \(\bm\mu=0\), this is simply
	\(\bm u_h(T)=\widetilde{\bm u}_h(T)\).
	For a queried point \(\bm x_a\), let \(\bm j(a)\) be the corresponding grid
	index.  If \(\bm\mu\neq0\), \(\bm j(a)\) should be replaced by the grid index
	corresponding to \(\bm x_a+\bm\mu T\) on the torus.

Define the point probability in the normalized state by
	\[
	p_a:=\bigl|(\ket{\bm u_h(T)})_{j(a)}\bigr|^2,
	\qquad
	Q_a:=p_a\|\bm u_h(T)\|^2
	=
	|(\bm u_h(T))_{j(a)}|^2 .
	\]
	Then the Cole-Hopf point-value readout is
	\begin{equation}\label{eq:HJ-point-readout}
		S_{\nu,h}(T,\bm x_a)
		=
		-\nu\log Q_a .
	\end{equation}
	Thus the readout requires estimating both the point probability \(p_a\) in the
	normalized state and the normalization factor \(\|\bm u_h(T)\|^2\).

	Let
	\[
	s_a:=|S_{\nu,h}(T,\bm x_a)|
	\]
	and assume \(s_a>0\).  We require the physical point-value accuracy
	\[
	\frac{
		|\widehat S_a-S(T,\bm x_a)|
	}{s_a}
	\lesssim \delta .
	\]
	The deterministic error has two components: the vanishing-viscosity error and
	the spatial discretization error.  We use the model
	\begin{equation}\label{eq:HJ-visc-error}
		\|S^\nu(T,\cdot)-S(T,\cdot)\|_{L^\infty}
		\le
		C_{\rm vv}D_{\rm vv}(d)\nu^\beta,
		\qquad
		\frac12\le \beta\le 1,
	\end{equation}
	where \(D_{\rm vv}(d)\) records the dimension dependence of the available
	vanishing-viscosity estimate.  In the conservative scaling used in
	\cite{JinLiu25HJ}, one may take \(D_{\rm vv}(d)=d\), with worst-case rate
	\(\beta=1/2\).  For the second-order periodic finite-difference
	discretization above, we use the spatial error model \(O(dh^2)\).  Hence it is
	sufficient to choose
	\begin{equation}\label{eq:nu-h-delta-HJ}
		\nu
		=
		\Theta\!\left(
		\left(
		\frac{\delta s_a}{D_{\rm vv}(d)}
		\right)^{1/\beta}
		\right),
		\qquad
		h
		=
		\Theta\!\left(
		\sqrt{\frac{\delta s_a}{d}}
		\right),
	\end{equation}
	up to problem-dependent constants.

	It remains to specify the readout accuracy.  If
	\[
	\widehat Q_a=Q_a(1+\eta_Q),
	\]
	then for \(|\eta_Q|\le1/2\),
	\[
	|\widehat S_{\nu,h}(T,\bm x_a)-S_{\nu,h}(T,\bm x_a)|
	=
	\nu|\log(1+\eta_Q)|
	\lesssim
	\nu|\eta_Q|.
	\]
	Therefore an \(O(\delta)\) relative readout error in
	\(S_{\nu,h}(T,\bm x_a)\) is ensured by
	\begin{equation}\label{eq:Q-relative-accuracy}
		|\eta_Q|
		=
		O\!\left(
		\min\left\{
		\frac12,\frac{\delta s_a}{\nu}
		\right\}
		\right).
	\end{equation}
	The fact that the admissible relative error in \(Q_a\) can become larger when
	\(\nu\) is small is not a contradiction: the Cole-Hopf map multiplies the
	logarithmic error by \(\nu\).  The difficulty for small \(\nu\) is instead
	reflected in the point probability \(p_a\), which may be exponentially small.

	Substituting \eqref{eq:nu-h-delta-HJ} into
	Corollary~\ref{cor:HJ-complexity}, and choosing the state-preparation accuracy
	small enough so that its induced additive error in \(p_a\) is negligible
	compared with \(p_a\eta_Q\), gives the matrix-query cost per successful
	preparation of the normalized surrogate state as
	\begin{equation}\label{eq:HJ-delta-complexity-general}
		C_{\rm prep}
		=
		\widetilde{O}\!\left(
		u_r
		\left(
		d\,D_{\rm vv}(d)^{-\frac{1}{2\beta}}
		(\delta s_a)^{\frac{1}{2\beta}-\frac12}
		\sqrt T
		+
		1
		\right)
		\right),
	\end{equation}
	where polylogarithmic factors in \(u_r\), \(1/\delta\), and the readout-related
	precision parameters are suppressed.  In the conservative scaling
	\(D_{\rm vv}(d)=d\), this becomes
	\begin{equation}\label{eq:HJ-delta-complexity}
		C_{\rm prep}
		=
		\widetilde{O}\!\left(
		u_r
		\left(
		d^{1-\frac{1}{2\beta}}
		(\delta s_a)^{\frac{1}{2\beta}-\frac12}
		\sqrt T
		+
		1
		\right)
		\right).
	\end{equation}
	This is only the state-preparation cost.  It is not yet the full cost of
	estimating \(S(T,\bm x_a)\), since the Cole-Hopf readout also requires
	estimating \(p_a\) and \(\|\bm u_h(T)\|^2\).

	With direct sampling, the point-probability estimation requires
	\[
	N_{\rm read}
	=
	\widetilde O\!\left(
	\frac{1}{p_a\eta_Q^2}
	\right)
	\]
	state preparations, up to confidence logarithms and the analogous cost for
	estimating the normalization factor \(\|\bm{u}_h(T)\|^2\).  Therefore the total
	direct-sampling cost scales as
	\[
	C_{\rm total}
	=
	\widetilde O\!\left(
	N_{\rm read}\,C_{\rm prep}
	\right).
	\]
	Amplitude estimation can improve the sampling dependence, but the readout cost
	remains a separate factor from the state-preparation complexity.

	\subsection{Biharmonic diffusion equation}
	\label{subsubsec:biharmonic}

	We next consider the biharmonic diffusion equation
	\[
	\partial_t u = -\Delta^2 u + f,\qquad u(0,x)=u_0(\bm x),\qquad \bm x\in\Omega\subset\mathbb{R}^d.
	\]
	We impose boundary conditions so that $\Delta^2$ is self-adjoint and nonnegative on $L^2(\Omega)$,
	e.g., the simply supported plate conditions
	\[
	u=0,\qquad \Delta u=0 \quad\text{on }\partial\Omega.
	\]
	Then $\mathcal{A}:=(-\Delta)^2=\mathcal{L}^\dagger\mathcal{L}$ with $\mathcal{L}:=-\Delta$. In the
	first-order embedding $\bm{\psi}=[\,w\;\;v\,]^\top$ in \eqref{eq:def psi-wv}--\eqref{eq:p}, we impose
	the homogeneous boundary conditions
	\[
	w(s,\cdot)=0,\qquad v(s,\cdot)=0 \quad\text{on }\partial\Omega .
	\]

	We discretize $-\Delta$ by standard second-order finite differences on a uniform grid with mesh size
	$h=1/N_x$. In one dimension, collecting interior nodal values into
	$\bm{w}_h(s),\bm{v}_h(s)\in\mathbb{R}^{N_x-1}$ and enforcing homogeneous Dirichlet data, we obtain
	\[
	-(\Delta w)_i \approx (B_h\bm w_h)_i,
	\qquad
	-(\Delta v)_i \approx ((B_h)^\top\bm v_h)_i,
	\quad i=1,\ldots,N_x-1,
	\]
	where $B_h$ denotes the standard second-order finite-difference negative Laplacian.

	In $d$ dimensions on $\Omega=(0,1)^d$, letting $M_w=(N_x-1)^d$ and identifying
	$\bm{w}_h(s),\bm{v}_h(s)\in\mathbb{R}^{M_w}$, we define the discrete Laplacian by the Kronecker sum
	\[
	L:=\sum_{k=1}^d
	\Bigl(I_{N_x-1}^{\otimes (k-1)}\otimes B_h\otimes I_{N_x-1}^{\otimes (d-k)}\Bigr),
	\]
	so that $L$ is symmetric positive definite and approximates $-\Delta$ with homogeneous Dirichlet data.

	\begin{corollary}\label{cor:biharmonic-complexity}
		Consider the biharmonic diffusion equation on $\Omega=(0,1)^d$ with mesh size
		$h=1/N_x$ and the discretization described above. Let $L$ denote the discrete Laplacian.
        Assume that the
		hypotheses of Theorem~\ref{thm:complexity} hold. Then  there is a quantum
		algorithm to prepare an $\varepsilon$-approximation
		of the normalized discrete solution state $\ket{\bm{u}_h(T)}$ with
		$\Omega(1)$ success probability, using
		\[
		\widetilde{O}\!\left(
		u_r\left(
		\frac{d}{h^2}\sqrt{T\,\log\!\frac{u_r}{\varepsilon}}
		+
		\log\!\frac{u_r}{\varepsilon}
		\right)\right)
		\]
		queries to the block-encoding oracle for $L$, where
		\[
		u_r=
		O\!\left(
		\frac{\|u_0\|_{L^2(\Omega)}+T\|f\|_{L^2(\Omega)}}{\|u(T)\|_{L^2(\Omega)}}
		\right).
		\]
		Moreover, the number of calls to the state-preparation oracles for
		$\bm{\psi}_0$ and $\bm{b}$ is $O(u_r)$.
	\end{corollary}
	\begin{proof}
		For the standard second-order finite-difference discretization on
		$\Omega=(0,1)^d$, the discrete Laplacian satisfies $\|L\|=O(d/h^2)$.
		Substituting this bound into Theorem~\ref{thm:complexity} yields the stated
		complexity estimate.
	\end{proof}

	\section{Long-time simulation for dissipative dynamics and linear algebra problems}
	\label{sec:longtime}

The complexity estimate in Theorem~\ref{thm:complexity} depends explicitly on
the final time \(T\). In the strictly dissipative case, this dependence can be
removed in the long-time regime because the solution converges exponentially to
its steady state. Thus one only needs to simulate up to the relaxation time,
rather than up to the physical final time. This observation also leads to a
quantum linear solver by identifying the steady state with the solution of
\(A\bm x=\bm b\).

	\subsection{Long-time simulation}
	Assume that $\mathcal{A}=\mathcal{L}^\dagger\mathcal{L}$ is self-adjoint and has a spectral gap
	$\sigma(\mathcal{A})\subset[\lambda_0,\infty)$ for some $\lambda_0>0$. Then the semigroup decays
	exponentially,
	\begin{equation*}
	\bigl\|e^{-\mathcal{A}T}\bigr\|_{\mathrm{L}(L^2(\Omega),L^2(\Omega))}
	:=\sup_{v\in L^2(\Omega)}\frac{\|e^{-\mathcal{A}T}v\|_{L^2(\Omega)}}{\|v\|_{L^2(\Omega)}}
	\le e^{-\lambda_0 T},
	\end{equation*}
	and by Duhamel's principle
	\begin{equation}\label{eq:longtime 1}
	u(T)=e^{-\mathcal{A}T}u_0+\int_0^T e^{-\mathcal{A}(T-s)}f\,\mathrm{d}s
	=\mathcal{A}^{-1}f+e^{-\mathcal{A}T}\bigl(u_0-\mathcal{A}^{-1}f\bigr).
	\end{equation}
	Consequently,
	\begin{equation}\label{eq:longtime 2}
	\|u(T)-\mathcal{A}^{-1}f\|_{L^2(\Omega)}
	\le e^{-\lambda_0 T}\,\|u_0-\mathcal{A}^{-1}f\|_{L^2(\Omega)}.
	\end{equation}
	Given any target accuracy $\varepsilon>0$, it is sufficient to simulate the dynamics up to the
	truncated time
	\[
	\tilde T:=\frac{1}{\lambda_0}\log\frac{C}{\varepsilon},
	\]
	for a problem-dependent constant $C\ge \|u_0-\mathcal{A}^{-1}f\|_{L^2(\Omega)}/\|\mathcal{A}^{-1}f\|_{L^2(\Omega)}$, in order to ensure
	\[\|u(T)-\mathcal{A}^{-1}f\|_{L^2(\Omega)}/\|\mathcal{A}^{-1}f\|\le \varepsilon,\quad \text{for\; all\;}\; T\ge \tilde T.\]
	In particular,
	the contribution of the initial datum is exponentially suppressed, and one may apply
	Theorem~\ref{thm:complexity} with the time parameter set to $\tilde T$, so that the resulting
	query complexity depends on $\tilde T$ rather than on the physical final time~$T$.

	Combining this observation with Theorem~\ref{thm:complexity}, we obtain the query
	complexity of our quantum algorithm in
	the long-time regime $T\gg \lambda_0^{-1}\log(1/\varepsilon)$.

	\begin{corollary}\label{cor:steady-complexity}
	Assume that the spatial discretization of \(\mathcal A\) gives a Hermitian
	positive definite matrix
	\[
	A=L^\dagger L,\qquad
	\sigma(A)\subset[\lambda_0,\infty),
	\qquad 0<\lambda_0 ,
	\]
	after the normalization of the block-encoding, and that the quantum access
	assumptions of Theorem~\ref{thm:complexity} hold.
		 Consider the discrete problem
		\[
		\frac{\mathrm{d} \bm u}{\mathrm{d} t}=-A\bm u(t)+\bm b,
		\qquad
		\bm u(0)=\bm u_0 .
		\]
		Let \(\bm u_\infty=A^{-1}\bm b\neq0\), and let \(T\ge \tilde T\), where
		\[
		\tilde T
		=
		\lambda_0^{-1}\log(C_{\rm lt}/\varepsilon),
		\qquad
		C_{\rm lt}\ge
		8\max\left\{
		1,\|\bm u_0-\bm u_\infty\|/\|\bm u_\infty\|
		\right\}.
		\]
	Then there exists a quantum algorithm that prepares an
		\(O(\varepsilon)\)-approximation of the normalized long-time solution state
		\(|\bm u(T)\rangle\) with \(\Omega(1)\) success probability.

		More precisely, with
		\[
		u_r^{\rm lt}
		:=
		\frac{\|\bm u_0\|+\tilde T\|\bm b\|}
		{\|\bm u(\tilde T)\|},
		\]
		the total number of queries to the block-encoding oracle for \(L\) is
		\[
		\widetilde{O}\!\left(
		u_r^{\rm lt}
		\left(
		\sqrt{
		\frac{\log(C_{\rm lt}/\varepsilon)}{\lambda_0} \|A\| \log\frac{u_r^{\rm lt}}{\varepsilon}}
		+
		\log\frac{u_r^{\rm lt}}{\varepsilon}
		\right)
		\right),
		\]
		and the number of calls to the state-preparation oracles is
		$
		O(u_r^{\rm lt}).
		$
	\end{corollary}

	\begin{remark}
		Under the same dissipativity and spectral-gap assumptions, the LCHS-based
		long-time algorithm in~\cite{YOA25} also achieves a query complexity
		independent of the physical final time \(T\), with
		\(\widetilde{O}(\log^3(1/\varepsilon))\) dependence on the target
		accuracy. In comparison, for fixed \(\lambda_0\) and fixed data ratio
		\(C_{\rm lt}\), Corollary~\ref{cor:steady-complexity} gives a leading
		\(\log^2(1/\varepsilon)\) dependence, up to the polylogarithmic overheads in
		Theorem~\ref{thm:complexity}. This improvement is due to the Gaussian tail of
		the Kannai kernel.
	\end{remark}

	\subsection{Linear algebra problems}
	We now interpret the long-time reduction as a linear systems solver for
	Hermitian positive definite problems
	\[A \bm x = \bm b, \quad A=L^\dagger L.\]

	Relabeling the normalized operator again by \(A\), we may assume
\[
\sigma(A)\subset[1,\kappa],
\qquad
\lambda_{\min}(A)=1,
\qquad
\|A\|=\kappa .
\]
Thus the spectral gap in the long-time estimate is \(\lambda_0=1\), and the
required relaxation time is
\[
\widetilde T=O\!\left(\log\frac1\varepsilon\right),
\]
when starting from zero. Combining this normalization with
Corollary~\ref{cor:steady-complexity} yields the following corollary.

\begin{corollary}
\label{cor:lin-sys}
Let \(A\in\mathbb C^{N\times N}\) be Hermitian positive definite and admit a
factorization \(A=L^\dagger L\). After normalization, assume that
\(\sigma(A)\subset[1,\kappa]\), and assume the hypotheses of
Theorem~\ref{thm:complexity}. Then, for any \(0<\varepsilon<1\), there exists
a quantum algorithm that, given a block-encoding of \(L\) and a
state-preparation oracle for \(\bm b\), outputs with success probability
\(\Omega(1)\) a normalized state \(\ket{\bm{x}^{\mathrm{a}}}\) satisfying
\[
\bigl\|
\ket{\bm{x}^{\mathrm{a}}}-\ket{\bm{x}}
\bigr\|
\le \varepsilon,
\qquad
\ket{\bm{x}}
=
\frac{A^{-1}\bm b}{\|A^{-1}\bm b\|}.
\]
The number of queries to the block-encoding oracle for \(L\) is
\[
\widetilde{O}\!\left(
\frac{\|\bm{b}\|}{\|\bm x\|} \log\frac{1}{\varepsilon}
\sqrt{\kappa\log\frac{\|\bm b\|}{ \varepsilon\|\bm x\| }\log \frac{1}{\varepsilon}}
\right).
\]
The number of calls to the state-preparation oracle for \(\bm b\) is
$O\!\left( \frac{\|\bm b\|}{\|\bm x\|} \log \frac{1}{\varepsilon} \right)$.
\end{corollary}

\begin{remark}
	The factor \(\|\bm b\|/\|\bm x\|\) in Corollary~\ref{cor:lin-sys} is
	data-dependent and can be much smaller than its worst-case upper bound. For
	linear systems obtained from PDE discretizations, this factor is often
	\(O(1)\) for smooth data. For example, consider the Poisson problem
	\[
	-\Delta u=f
	\]
	on a fixed bounded domain, discretized on a quasi-uniform mesh. If \(u\) and
	\(f\) are sufficiently smooth and \(\|u\|_{L^2(\Omega)}\) is bounded away from
	zero, then, under the usual grid-value representation, or equivalently an
	\(L^2\)-consistent mass-scaled finite element representation,
	\[
	h^{d/2}\|\bm x\| \asymp \|u\|_{L^2(\Omega)},
	\qquad
	h^{d/2}\|\bm b\| \asymp \|f\|_{L^2(\Omega)} .
	\]
	Hence
	$
	\frac{\|\bm b\|}{\|\bm x\|}
	\asymp
	\frac{\|f\|_{L^2(\Omega)}}{\|u\|_{L^2(\Omega)}}
	=O(1),
	$
	with constants independent of the mesh size. In this common smooth-data
	regime, the query complexity in Corollary~\ref{cor:lin-sys} reduces to
	\[
	\widetilde O\!\left(
	\sqrt{\kappa}\,\log^2\frac1\varepsilon
	\right),
	\]
	up to logarithmic factors in \(\kappa\).

	Without such regularity or low-frequency structure, one only has the general
	bound
	\[
	\frac{\|\bm b\|}{\|\bm x\|}
	=
	\frac{\|A\bm x\|}{\|\bm x\|}
	\le \|A\|
	=\kappa,
	\]
	under the normalization \(\sigma(A)\subset[1,\kappa]\). This gives the
	worst-case estimate
	\[
	\widetilde O\!\left(
	\kappa^{3/2}\log^2\frac{\kappa}{\varepsilon}
	\right).
	\]
Further improvement of the worst-case condition-number dependence could be
obtained by incorporating additional techniques, such as variable-time
amplitude amplification (VTAA) or optimized state-conversion procedures from
quantum linear-system algorithms~\cite{Ambainis12,CKS17,SSO19}. Such refinements
are largely orthogonal to the long-time transmutation mechanism used here.
\end{remark}

\begin{remark}\label{rem:laplace-qevt}
	The long-time linear solver is also related to the
	Laplace-transform-based quantum eigenvalue transformation (QET) framework of
	An et al.~\cite{ACLY26}. There, matrix functions admitting a Laplace
	representation are written as
	\[
	h(A)=\int_0^\infty g(t)e^{-tA}\,\mathrm dt .
	\]
	Thus inverse powers \(A^{-k}\) are natural examples, and our direct
	long-time Kannai solver corresponds to the special case
	\(h(z)=z^{-1}\), \(g(t)=1\).

	The main distinction is structural. The Laplace-transform QET framework is
	more general and implements \(e^{-tA}\) by LCHS without requiring an explicit
	factorization \(A=L^\dagger L\). In contrast, our construction exploits this
	factorization and represents \(e^{-tA}\) through the Kannai transform, namely
	as a Gaussian average of unitary wave evolutions generated by \(L\). This is
	what gives the factorization-induced \(\sqrt{t}\)-type auxiliary-time scaling.
\end{remark}

	\section{Other transmutation methods for quantum algorithms}\label{sec:other_transmutation}
	We collect further examples where a non-unitary evolution can be represented as a superposition of
	reversible dynamics, consistent with the transmutation viewpoint.

	Let $\mathcal{A}$ be a linear operator and consider two evolution families $w(s,\mathcal{A})$ and
	$u(t,\mathcal{A})$. Following \cite{Hersh74Transmutation}, a transmutation from the $s$-dynamics to
	the $t$-dynamics is an identity of the form
	\begin{equation}\label{eq:transmutation-abstract}
		u(t,\mathcal{A})
		=\int_{\Gamma} K(t,s)\,w(s,\mathcal{A})\,\mathrm{d}s,
	\end{equation}
	where the kernel $K(t,s)$ is independent of the particular realization of $\mathcal{A}$.
	In the examples below, $w(s,\mathcal{A})$ is unitary, or can be obtained from a unitary dynamics.
	The formula \eqref{eq:transmutation-abstract} then yields a continuous linear combination of unitary
	evolutions applied to the initial data.

	\subsection{Spherical means and the Euler--Poisson--Darboux equation}
	A classical example is the Euler--Poisson--Darboux equation \cite{courantHilbertII}
	\[
	u_{tt}+\frac{d-1}{t}u_t=\Delta_x u,
	\qquad
	u(0,x)=u_0(x),\quad u_t(0,x)=0.
	\]
	If $w$ solves the $d$-dimensional wave equation
	\[
	w_{tt}=\Delta_x w,
	\qquad
	w(0,x)=u_0(x),\quad w_t(0,x)=0,
	\]
	then the transmutation formula reads
	\begin{equation}\label{eq:EPD-transmutation}
		u(t,x)
		= c_d \int_{-1}^{1} w(\lambda t ,x)\,
		\bigl(1-\lambda^2\bigr)^{\frac{d-3}{2}}\,\mathrm{d}\lambda,
	\end{equation}
	which is a special case of the Delsarte--Lions transmutation \cite{Lions1956,Lions1959Weinstein}.

	Introducing the first-order wave formulation similar to \eqref{eq:def psi-wv}--\eqref{eq:unitary dynamics}, one has
	\[
	w(\cdot,s)=\Pi_1\mathcal{U}(s)\bm{\psi}_0,
	\qquad
	\bm{\psi}_0=[\,u_0\;\;0\,]^\top,
	\]
	where $\mathcal{U}(s)=e^{s\widetilde{\mathcal{L}}}$ with
	$\widetilde{\mathcal{L}} = -\ket{0}\bra{1}\otimes \mathcal{L}^\dagger + \ket{1}\bra{0}\otimes \mathcal{L}$.
	Then \eqref{eq:EPD-transmutation} becomes
	\[
	u(t,\cdot)
	=\frac{c_d}{t}\,\Pi_1\int_{-t}^{t}\Bigl(1-\frac{s^2}{t^2}\Bigr)^{\frac{d-3}{2}}
	\mathcal{U}(s)\bm{\psi}_0\,\mathrm{d}s .
	\]
	Approximating the integral by a quadrature with weights $\{c_j\}_{j=1}^M$ and nodes
	$\{s_j\}_{j=1}^M\subset[-t,t]$ yields an LCU form
	$\sum_{j=1}^M c_j\,\mathcal{U}(s_j)\bm{\psi}_0$ with coefficient $\ell_1$-norm
	$\alpha_c=\sum_j|c_j|\approx \int_{-t}^{t}|K(t,s)|\,ds$.
	Then we have
	\[
	\int_{-t}^{t}|K(t,s)|\,\mathrm{d} s
	=\frac{c_d}{t} \int_{-t}^{t} (1-\frac{s^2}{t^2})^{\frac{d-3}{2}}\;\mathrm{d} s
	= O(c_d),
	\qquad\text{hence}\qquad
	\alpha_c = O\big(c_d\big).
	\]
	Therefore, for any target precision $\varepsilon\in(0,1)$, there exists a quantum algorithm
	that prepares an $\varepsilon$-approximation of the normalized state proportional to $u(T)$
	with $\Omega(1)$ success probability and a success flag, using
	\[
	\widetilde{O}\left(
	\,\frac{ c_d\|u_0\|_{L^2(\Omega)}}{\|u(T)\|_{L^2(\Omega)}}
	\|L\|T
	\right)
	\]
	queries to a block-encoding of the first-order wave generator associated with \(L\).
	\begin{remark}
		This transmutation perspective has two immediate advantages for quantum algorithms.
		First, it represents the non-unitary dynamics associated with the second-order-in-time equation as a
		continuous superposition of unitary wave propagators.
		Second, it removes the explicit time dependence in the coefficient $(d-1)/t$ by transferring it to
		the integration kernel, while the underlying generator $\widetilde{\mathcal{L}}$ is time independent.
	\end{remark}

	\subsection{Transport--heat averaging}
	Following \cite{LuZuazua2016}, consider the transport equation on $\mathbb{R}^d$,
	\[
	y_t+\alpha\cdot\nabla_x y=0,
	\qquad
	y(0,x;\alpha)=y_0(x),
	\]
	where $\alpha$ is a standard Gaussian random vector with density
	$K(\alpha)=(2\pi)^{-d/2}e^{-|\alpha|^2/2}$. The solution is
	$y(t,x;\alpha)=y_0(x-t\alpha)$ and its average is
	\[
	\tilde y(t,x)
	=\int_{\mathbb{R}^d}K(\alpha)\,y_0(x-t\alpha)\,\mathrm{d}\alpha.
	\]
	One checks that $\tilde y$ satisfies $\tilde y_t=t\,\Delta_x \tilde y$, with $
	\tilde y(0,\cdot)=y_0$.
	Defining $u(t,x):=\tilde y\bigl(\sqrt{2t},x\bigr)$ yields the heat equation
	\[
	u_t-\Delta_x u=0,
	\qquad
	u(0,\cdot)=y_0.
	\]
	In operator form, set $\mathcal{T}(\alpha):=-\alpha\cdot\nabla_x$. Then
	\begin{equation}\label{eq:u-from-transport-avg}
		u(t,\cdot)=\int_{\mathbb{R}^d}K(\alpha)\,e^{\sqrt{2t}\,\mathcal{T}(\alpha)}\,y_0\,\mathrm{d}\alpha,
	\end{equation}
	and $e^{s\mathcal{T}(\alpha)}$ is unitary on $L^2(\mathbb{R}^d)$.

	For the periodic setting discretized by a Fourier spectral method, the discrete generator is
	skew-Hermitian and can be written as $\mathcal{T}_h(\alpha)=-iH(\alpha)$ with $H(\alpha)$ Hermitian.
	Truncating the Gaussian measure to $|\alpha|\le R_\alpha$ with
	$R_\alpha=O(\sqrt{\log(1/\varepsilon)})$ yields $\|H(\alpha)\|=O(R_\alpha\|H_0\|)$ for a
	fixed Hermitian matrix $H_0$ determined by the spatial grid. The resulting quantum algorithm has
	query complexity
	\[
	\widetilde{O}\!\left(
	\frac{\|u_0\|_{L^2(\Omega)}}{\|u(T)\|_{L^2(\Omega)}}\,\|H_0\|\sqrt{T\log\frac{1}{\varepsilon}}
	\right),
	\]
	with polylogarithmic factors suppressed.

	\begin{remark}
		For the periodic heat equation, the transport--heat averaging representation
		\eqref{eq:u-from-transport-avg} yields a quantum algorithm whose matrix-query
		complexity matches that of the Kannai-transform-based construction at the level
		of leading-order scaling in $T$ and $\varepsilon$.
		However, this approach relies on the discrete transport generator
		$\mathcal{T}_h(\alpha)$ being anti-Hermitian so that
		$\exp\bigl(s\,\mathcal{T}_h(\alpha)\bigr)$ is unitary.
		This is naturally satisfied by Fourier spectral discretizations under periodic
		boundary conditions, but it is incompatible with standard upwind schemes for
		first-order transport, which are dissipative and therefore do not yield an
		anti-Hermitian discrete operator.
		Moreover, extending the transport--averaging construction beyond periodic
		settings is nontrivial, since other boundary conditions and general domains do
		not preserve the simple translation structure underlying
		\eqref{eq:u-from-transport-avg}.
	\end{remark}

	\section{Numerical tests}\label{sec:numerical-tests}

	In this section we present numerical experiments to validate (i) the correctness of the
	Kannai transform used to recover the heat solution from a wave equation and
	(ii) the correctness of the unitary embedding obtained via the first-order factorization.
	We consider the one-dimensional heat equation on $\Omega=[0,1]$,
	\begin{equation}\label{eq:heat-1d}
		\partial_t u(t,x)=\partial_{xx}u(t,x),\qquad u(0,x)=u_0(x),
	\end{equation}
	with two different boundary conditions. In both cases we compute the solution at $T=1$ and
	compare the Kannai-based solution with a classical finite-difference (FD) time-marching solution.

	All results below are obtained via classical simulation and serve to validate (i) and (ii), rather than to demonstrate quantum speedup.
	We evaluate the truncated Kannai Gaussian integral on $[-R,R]$ with $R=10$ using the trapezoidal rule with $M=800$ panels, and implement the unitary propagator $e^{-iHs}$ in MATLAB via \texttt{expm} (i.e., \texttt{expm(-1i*H*s)}).

	\noindent\textbf{Case 1: inhomogeneous Dirichlet boundary condition.}
	We set
	\[
	u(t,0)=u(t,1)=1,
	\qquad
	u_0(x)=\cos(2\pi x).
	\]
	Using the staggered discretization with \(h=1/N_x\) and \(N_x=50\), let
	\(L:=D_h\), where \(D_h\) is defined in~\eqref{eq:L-1D}. For the prescribed
	boundary values, the boundary correction is
	\[
	\bm g_{\rm bd}^{(1)}
	=
	\frac1h(1,0,\ldots,0,-1)^\top\in\mathbb R^{N_x}.
	\]
	Since the physical source is zero in this test, the effective source is
	$
	\bm f_h^{\rm eff}
	=
	- L^\dagger \bm g_{\rm bd}^{(1)}
	\in\mathbb R^{N_x-1}$.
	The embedded first-order system is
	\[
	\frac{\mathrm{d}}{\mathrm{d} s}
	\begin{bmatrix}
		\bm w_h\\
		\bm v_h
	\end{bmatrix}
	=
	\begin{bmatrix}
		0 & L^\dagger\\
		-L & 0
	\end{bmatrix}
	\begin{bmatrix}
		\bm w_h\\
		\bm v_h
	\end{bmatrix}
	+s
	\begin{bmatrix}
		\bm f_h^{\rm eff}\\
		0
	\end{bmatrix}
	=
	-iH
	\begin{bmatrix}
		\bm w_h\\
		\bm v_h
	\end{bmatrix}
	+s\bm b,
	\qquad
	H=i
	\begin{bmatrix}
		0 & L^\dagger\\
		-L & 0
	\end{bmatrix}.
	\]
	Let
	\[
	\bm\psi_0=
	\begin{bmatrix}
		\bm u_{0,h}\\
		0
	\end{bmatrix},
	\qquad
	\Pi_w=
	\begin{bmatrix}
		I_{N_x-1} & 0
	\end{bmatrix},
	\qquad
	U(s)=e^{-iHs}.
	\]
	We approximate the interior heat solution at time \(T\) by the Kannai
	quadrature
	\[
	\bm u_{h}(T)
	=
	\Pi_w
	\sum_{m=0}^{M}\omega_m
	\Bigl(
	\kappa_T(s_m)U(s_m)\bm\psi_0
	+
	\Lambda_T(s_m)U(s_m)\bm b
	\Bigr),
	\]
	where \((s_m,\omega_m)\) are trapezoidal nodes and weights on \([-R,R]\), and
	\[
	\kappa_T(s)=(4\pi T)^{-1/2}e^{-s^2/(4T)},
	\qquad
	\Lambda_T(\sigma)
	=
	\sqrt{\frac{T}{\pi}}e^{-\sigma^2/(4T)}
	-
	\frac{|\sigma|}{2}
	\operatorname{erfc}\!\left(\frac{|\sigma|}{2\sqrt T}\right).
	\]
	Figure~\ref{fig:numerical-tests} (left) compares the resulting interior values
	with a standard finite-difference time-marching solution.

	\noindent\textbf{Case 2: homogeneous Neumann boundary condition.}
	We impose \[\partial_x u(t,0)=\partial_x u(t,1)=0,\quad u_0(x)=\cos(\pi x).\]
    With the staggered placement described in Remark~\ref{remark:Neumann bd},
	the semi-discrete Neumann problem takes the form of a homogeneous first-order
	system. Applying the same Kannai-transform quadrature as above gives the result
	shown in Fig.~\ref{fig:numerical-tests} (right).

	In both boundary-condition settings, the numerical solutions obtained from the Kannai transform
	agree with the reference finite-difference solutions, which supports the correctness of
	(i) and (ii).

 \begin{figure}[t]
 	\centering
 	\begin{minipage}{0.48\textwidth}
 		\centering
 		\includegraphics[width=\linewidth]{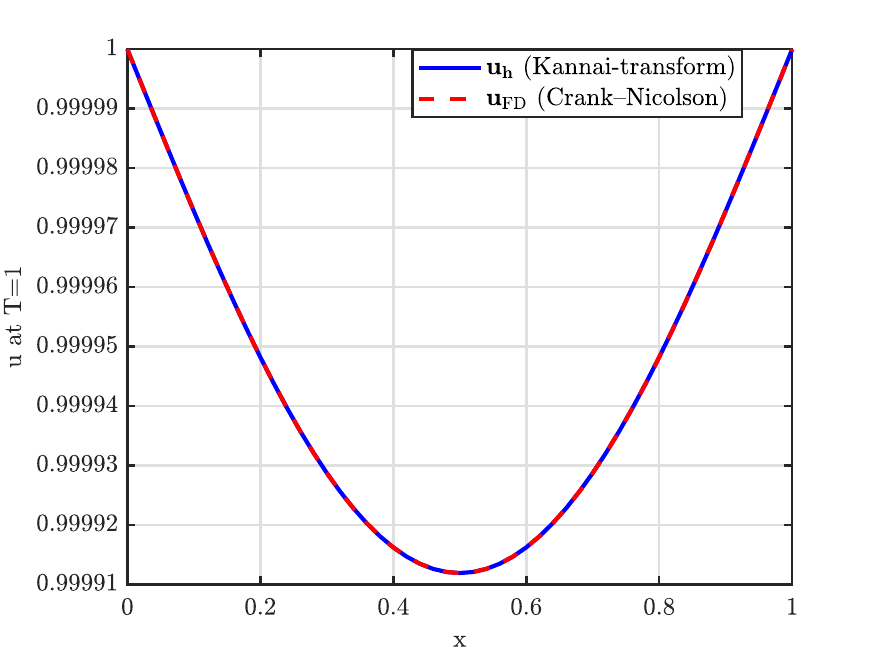}
 		\par\smallskip{\footnotesize\itshape (a) Case 1: Dirichlet $u(t,0)=u(t,1)=1$.\par}
 	\end{minipage}\hfill
 	\begin{minipage}{0.48\textwidth}
 		\centering
 		\includegraphics[width=\linewidth]{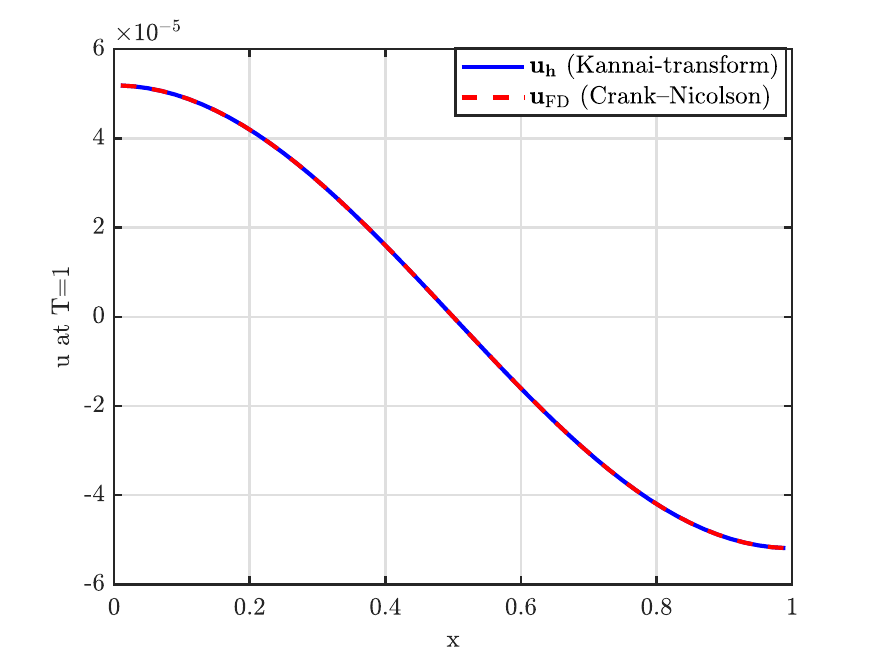}
 		\par\smallskip{\footnotesize\itshape (b) Case 2: Neumann $\partial_x u|_{\partial\Omega}=0$.\par}
 	\end{minipage}
 	\caption{Comparison between the Kannai-based numerical solution
 		and the classical finite-difference solution at $T=1$.}
 	\label{fig:numerical-tests}
 \end{figure}

	\section{Conclusion}\label{sec:conclusions}
    We developed a transmutation-based quantum framework for dissipative dynamics with factorized generators \(A=L^\dagger L\). By exploiting the Gaussian Kannai kernel, the method achieves improved query complexity with square-root dependence on \(T\|A\|\) and only polylogarithmic dependence on the target accuracy, up to the standard state-preparation and output-norm factors. We illustrated the framework for representative parabolic PDEs and showed that, in the long-time regime under a spectral-gap assumption, it also yields a structured quantum linear solver. This advantage is tied to the factorized access model, compatibility with the Gaussian Kannai transmutation, and favorable output norms; otherwise it may be reduced. Future work includes quantum preconditioning and extensions to time-dependent and more general non-self-adjoint operators.

	\section*{Acknowledgments}
	SJ acknowledges the support of the NSFC grant No. 12341104, the Shanghai Pilot Program for Basic Research,  the Science and Technology Commission of Shanghai Municipality (STCSM) grant no. 24LZ1401200, the Shanghai Jiao Tong University 2030 Initiative, and the Fundamental Research Funds for the Central Universities.
	CM was  supported by Fundamental and Interdisciplinary Disciplines Breakthrough Plan of the Ministry of Education of China (JYB2025XDXM112), NSFC grant No. 12501607, and
	the Science and Technology Commission of Shanghai Municipality (No.22DZ2229014).
    EZ was funded by the Alexander von Humboldt-Professorship program,
    the ERC Advanced Grant CoDeFeL, the Grant PID2023-146872OB-I00-DyCMaMod of MICIU (Spain),  the COST Actions  CA24136 - InterCoML and  CA24122--mSPACE, supported by COST (European Cooperation in Science and Technology),  the AFOSR 24IOE027 project, and the SURE-AI Centre grant 357482, Research Council of Norway.

\end{document}